\documentclass[11pt]{article}
\usepackage[letterpaper,hmargin=1in,vmargin=1in,footskip=36pt]{geometry}
\usepackage{booktabs} % For formal tables
\usepackage[ruled]{algorithm2e} % For algorithms

\usepackage{authblk}

\usepackage[utf8]{inputenc}
\usepackage[T1]{fontenc}
\usepackage[bookmarks=false,colorlinks=true, linkcolor=blue, urlcolor=blue, citecolor=blue, breaklinks=true]{hyperref}

\usepackage{kpfonts,lipsum}
\usepackage{empheq}
\usepackage{mathtools,adjustbox}				% bessere Variante von amsmath (das von diesem Paket geladen wird). Zahlreiche Mathesachen
\usepackage{amssymb,subfig}				% zahlreiche mathematische Symbole
\usepackage{mathrsfs, stmaryrd}		% Zusätzliche Mathesymbole
\usepackage{tikz}
\usepackage{amsthm}			% Theorem- und Beweisumgebungen
\usetikzlibrary{chains,shapes.multipart}
\usetikzlibrary{shapes,calc}
\usetikzlibrary{automata,positioning}
\usetikzlibrary{arrows, decorations.markings}
\definecolor{myred}{RGB}{220,43,25}
\definecolor{mygreen}{RGB}{0,146,64}
\definecolor{myblue}{RGB}{0,143,224}

\renewcommand{\vec}{}
\renewcommand{\P}{\mathbb{P}}

\DeclareMathOperator{\conv}{conv}
\DeclareMathOperator{\moag}{mag}

\newcommand{\Sc}{\mathcal{S}}
\newcommand{\Z}{\mathbb{Z}}
\renewcommand{\P}{\mathcal{P}}
\newcommand{\B}{\mathcal{B}}
\newcommand{\EB}{\mathcal{EB}}

\DeclareMathOperator{\relint}{relint}
\DeclareMathOperator{\val}{val}
\DeclareMathOperator{\relaxation}{R}

\DeclareMathOperator{\cost}{cost}

\DeclareMathOperator{\polymatroid}{polymatroid}
\DeclareMathOperator{\utility}{utility}

\newcount\colveccount
\newcommand*\colvec[1]{
        \global\colveccount#1
        \begin{pmatrix}
        \colvecnext
}
\def\colvecnext#1{
        #1
        \global\advance\colveccount-1
        \ifnum\colveccount>0
                \\
                \expandafter\colvecnext
        \else
                \end{pmatrix}
        \fi
}
\theoremstyle{definition}
\newtheorem{definition}{Definition}[section]
\newtheorem{example}[definition]{Example}

\theoremstyle{plain}
\newtheorem{assumption}[definition]{Assumption}
\newtheorem{proposition}[definition]{Proposition}
\newtheorem{lemma}[definition]{Lemma}
\newtheorem{corollary}[definition]{Corollary}
\newtheorem{theorem}[definition]{Theorem}

\theoremstyle{remark}
\newtheorem{remark}[definition]{Remark}

\usepackage{array}				
\usepackage{framed,bm}

\newcommand{\N}{\mathbb{N}}	% Natürliche Zahlen
\newcommand{\IQ}{\mathbb{Q}}	% Rationale Zahlen

\newcommand{\R}{\mathbb{R}}	% Reelle Zahlen
\DeclareMathOperator{\LP}{LP}

\renewcommand{\vec}{}
\renewcommand{\bm}{}

\title{Pricing in Resource Allocation Games Based on\\ Lagrangean Duality and Convexification}

\author{Tobias Harks }
\affil{\small Augsburg University, Institute of Mathematics, 86135 Augsburg\\
	\href{mailto:tobias.harks@math.uni-augsburg.de}{\texttt{tobias.harks@math.uni-augsburg.de}}}

\date{First version, July 2019, this version \today}
\begin{document}

\maketitle
\begin{abstract}
We consider a basic resource allocation game, where
the players' strategy spaces are subsets of $\R^m$ and
cost/utility functions  are parameterized
by some common vector $\vec u\in \R^m$
and, otherwise, only depend on the own strategy choice.
A strategy of a player can be interpreted as a vector of resource consumption 
and a joint strategy profile naturally leads to an aggregate
consumption vector. Resources can be priced, that is, the game is augmented by a price vector $\bm \lambda\in\R^m_+$
and  players have quasi-linear  overall costs/utilities 
meaning that in addition to the original costs/utilities, a player needs to pay the corresponding price per consumed 
unit. We investigate the following 
question:  for which aggregated consumption vectors $\vec u$ 
can we find prices $\bm \lambda$ that induce an equilibrium
realizing the targeted consumption profile?

For answering this question, we revisit a well-known duality-based framework and derive several characterizations of the existence of such $\vec u$ and $\lambda$ using convexification techniques. We show
that for finite strategy spaces or certain concave games, the equilibrium existence problem reduces to solving a well-structured LP.
We then consider a class of monotone aggregative games 
having the property that the cost/utility functions
of players may depend on the induced load of a strategy profile.
For this class, we show a sufficient condition of
enforceability based on the previous characterizations.
We demonstrate that this framework  can help to unify
 parts of four largely independent streams in the literature: tolls
in transportation systems, Walrasian market equilibria, trading networks and congestion control in communication networks.  Besides reproving existing
results we establish new existence results by using methods from
polyhedral combinatorics, polymatroid theory and discrete convexity.\end{abstract}

% Paper body
\section{Introduction}
Resource allocation problems appear in several real-world
situations.
Whenever available resources need to be matched to demands,
the goal is to find the most profitable or least costly allocation of
the resources. Applications can be found
in several areas, including traffic
and telecommunication networks.  In 
the above applications,  a
finite (or infinite) number of  players interact strategically, each optimizing their individual
objective function. The corresponding allocation of resources in such setting is usually determined by an equilibrium solution of the underlying strategic game. 
A central question in all these areas concerns the problem
of how to incentivize players in order to use the (scarce) resources
optimally. One key approach in all named application areas is the concept
of \emph{pricing resources} according to their usage.
Every resource comes with an anonymous prices per unit of consumption
and defining the ``right'' prices thus offers the chance of
inducing equilibria with optimal or efficient resource usage.
Prominent examples are toll pricing in transportation
networks, congestion pricing in telecommunication networks and
\emph{market pricing} in economics.
A prime example of the latter is the \emph{Walrasian competitive equilibrium} (cf. Walras~\cite{walras54}), where goods are priced such that
there is an allocation of goods to buyers  with the property that every buyer gets a bundle of items maximizing her overall utility
given the current prices for the goods.

In this paper, we will introduce a generic model of pricing in resource
allocation games with quasi-linear costs/utilities that subsumes several of the above mentioned applications
as a special case. In the following, we first introduce the model formally, discuss applications and 
then give an overview on the main results and related work.
\subsection{The Model}\label{sec:model}
 Let $E=\{1,\dots, m\}$ be a finite and non-empty set of resources
 and $N=\{1,\dots,n\}$ be a nonempty finite set of players. 
 For $i\in N$, let $X_i\subset \R^m, X_i\neq\emptyset$ denote the strategy space of player $i$ and
define $X:=\times_{i\in N} X_i$ as the combined strategy space.
The vector $\vec x_i=(x_{ij})_{j\in E} \in X_i $ is a strategy profile
of player $i\in N$ and the entry $x_{ij}\in \R$ can be interpreted  as the level of resource
usage of player $i$ for resource $j$.  
For every player $i\in N$, there is a function
$g_i:\R^m\rightarrow \R^m, \vec x_i\mapsto g_i(\vec x_i) $ mapping 
the resource usage vector to a vector of the actual \emph{resource consumption}. The function $g_i$ allows to model player-specific characteristics such as weights. For both $\vec x_i$ and $g_i$
negative values are allowed.
We  call the vector of resource usage $\vec x=(x_{ij})\in \R^{n\cdot m}$ a
 \emph{strategy distribution}.
Given $\vec x\in X$, we can define
the \emph{load} on resource $j\in E$ as
$\ell_j(\vec x):=\sum_{i\in N} g_{ij}(\vec x_{i}),$
where $g_{ij}$ is the $j$-th component of $g_i$.
In the following, we introduce properties
of utility functions needed for our main results.
We will distinguish between cost minimization games and
utility maximization games.

\begin{assumption}\label{ass:aggregate}
We assume that cost/utility functions
are parameterized by an \emph{exogenously given vector $\vec u\in \R^m$}  and depend on the
own strategy vector only.
\begin{enumerate}
\item For minimization games $G^{\min}(\vec u)$ with respect to $\vec u\in \R^m$, the total cost of a player $i\in N$ under strategy distribution $\vec x\in X$ is defined by a function $\cost_{i}:X\rightarrow \R,$
which satisfies 
\[ \cost_i(\vec x):= \pi_i(\vec u,\vec x_i) \text{ for all }\vec x\in X
\text{ for some function $\pi_i:\R^m\times X_i\rightarrow \R$.}\]
\item For maximization games $G^{\max}(\vec u)$, we denote the utility function 
for $i\in N$ by 
$\utility_{i}:X\rightarrow \R$ and we assume that it satisfies 
\[ \utility_i(\vec x):= v_i(\vec u,\vec x_i) \text{ for all }\vec x\in X
\text{ for some function $v_i:\R^m\times X_i\rightarrow \R$.}\]
\end{enumerate}
\end{assumption}
The vector $\vec u$ can be interpreted
as the induced load of an equilibrium, that is, 
$\vec u=\ell(\vec x)$. We assume for the moment that players are \emph{load taking} in the sense
that they assume not being able to influence the global load vector $\vec u$
by their own strategy $\vec x_i$, thus leading to the prescribed shape
of the cost/utility functions -- we will later also consider models in which a functional dependency of the strategy choice on the induced load is allowed.

%Clearly $x_j$ is a function of $\vec x$ but in order
%to keep notation simple we simply write $x_j$.

\subsection{Pricing in Resource Allocation Games}\label{sec:pricing}
We are concerned with the problem of defining \emph{prices} $\lambda_j \geq 0, j\in E$ on 
the resources
in order to incentivize an efficient usage of the resources
as explained below.
If player $i$ uses resource $j$
at consumption level $g_{ij}(\vec x_i)$, she needs to pay
$\lambda_j g_{ij}(\vec x_i)$. The quantities $\pi_i(\vec u,\vec x_i)$ and $\bm\lambda^\intercal g_{i}(\vec x_i)$ are assumed to be normalized
to represent the same unit (say money in Euro)
and we assume that the private cost functions are quasi-linear:
$ \pi_i(\vec u,\vec x_i)+\bm\lambda^\intercal g_{i}(\vec x_i).$
If the parameter $\vec u=(u_{j})_{j\in E}\in \R^m$ represents
a targeted load vector $\ell(\vec x^*)$,
then, the task is to find prices $\bm{\lambda}\in\R^m_+$ so that  $\vec x^*$
becomes an equilibrium of the game with prices.
\begin{definition}[Enforceability]\label{def:enforceable}
We now introduce three variants of enforceability.
\begin{enumerate}
\item\label{enum:enforceable} A vector $\vec u\in \R^m$ is enforceable,
if there is a tuple $(\vec x^*,\bm \lambda)\in X\times\R^m_+$ satisfying~\ref{cond1}. and~\ref{cond2}.
for minimization games $G^{\min}(\vec u)$ or~\ref{cond1}. and~\ref{cond3}.
for maximization games $G^{\max}(\vec u)$:
\begin{enumerate}
\item\label{cond1}  $\ell_j(\vec x^*)= u_j$ for all $j\in E$.
\item Minimization: \label{cond2} $ \vec x^*_i\in\arg\min_{\vec x_i\in X_i}\left\{ \pi_{i}(\vec u,\vec x_i)+\bm\lambda^\intercal g_i(\vec x_i)\right\} \text{ for all }i\in N.$
\item Maximization: \label{cond3} $ \vec x^*_i\in\arg\max_{\vec x_i\in X_i}\left\{ v_i(\vec u,\vec x_i)-\bm\lambda^\intercal g_i(\vec x_i)\right\} \text{ for all }i\in N.$
\end{enumerate}
In this case, we say $\vec u$ can be enforced by $(\vec x^*,\bm \lambda)\in X\times\R^m_+$.
\item\label{enum:weakly-enforceable-market}
A vector $\vec u\in \R^m$ is called \emph{weakly enforceable with market clearing prices}, if there is a tuple $(\vec x^*,\bm \lambda)$ 
that satisfies the above condition~\eqref{cond2} (or~\eqref{cond3})
but~\eqref{cond1} is replaced with $\ell(\vec x^*)\leq\vec u$
and $\bm \lambda$ satisfies the Walrasian law that
resources $j\in E$ with slack capacity have zero price, that is, $\ell_j(\vec x^*)<\vec u_j \Rightarrow \lambda_j=0$ for all $j\in E$.
\item\label{enum:unique-enforceable}
A vector $\vec u\in \R^m$ is uniquely enforceable,
if there is $\bm \lambda\in \R^m_+$ 
and a unique $\vec x^*\in X$  satisfying~\ref{cond1}. and~\ref{cond2}.
for minimization games $G^{\min}(\vec u)$ or~\ref{cond1}. and~\ref{cond3}.
for maximization games $G^{\max}(\vec u)$.
\end{enumerate}
 \end{definition}
Condition~\ref{cond1}. requires that $\vec x^*$ realizes the capacities $
\ell(\vec x^*)=\vec u$ 
while Condition~\ref{cond2}. implements $\vec x^*$ as a pure Nash equilibrium of the minimization game $G^{\min}(\vec u)$ augmented with prices. Condition~\ref{cond1}. and~\ref{cond3}.
refer to a maximization game $G^{\max}(\vec u)$ augmented with prices. 
The definition of weakly enforceable capacity vectors (with market prices)
is motivated by applications, for which outcomes are interesting that do not use all capacities at equilibrium.
\subsection{Running Examples}
We give four prototypical examples that
are used throughout the paper.
\begin{example}[Tolls in Network Routing]\label{ex:routing}
There is a directed
graph $G=(V,E)$ and a finite set $N$ of populations of commuters modeled
by tuples $(s_i,t_i,d_i), i\in N$, where $s_i$ is the source,
$t_i$ the sink and $d_i>0$ represents the volume
of flow that is traveling from $s_i$ to $t_i$.
In this setting, we can think of the set $E$
as being the set of resources and $X_i$
representing a flow polytope for every population $i\in N$.
In the network routing literature, there are several
equilibrium notions known according to whether
the underlying model is \emph{nonatomic}
(Wardrop equilibrium) or  \emph{atomic} (Nash equilibrium).
 Given an equilibrium concept, the goal is to find network tolls $\lambda_j\geq 0, j\in E$
on edges that enforce a prescribed capacity vector $\vec u$
via an equilibrium strategy distribution.
\iffalse
One important aspect of these two different models
is that in the atomic model, the set of flows
$X_i$ carries integrality requirements while the non-atomic
formulation usually involves convex sets $X_i$.
Let us emphasize that most models in the area of network routing assume that the cost of a player (or an agent)
only depends on the own strategy choice and the aggregate
load vector induced by the strategies of the competitors.
Moreover for nonatomic models, an agent cannot influence the load vector by her own strategy choice.
In this regard, only nonatomic games seem to fit into the class
of games $G^{\min}(\vec u)$ augmented with prices introduced in Assumption~\ref{ass:aggregate}.
But as we will see later, also \emph{atomic} congestion games can be handled.
\fi
\end{example}
Now we turn to the area
of \emph{Walrasian market equilibria} which constitutes a central
topic in the economics literature, see the original work of Walras~\cite{walras54} and later landmark papers of Kelso and Crawford~\cite{Kelso82}, Gul and Stachetti~\cite{Gul99}
and Danilov et al.~\cite{Danilov2001}.
\begin{example}[Market Equilibria]\label{ex:auction}
Suppose there are items $E=\{1,\dots,m\}$  for sale
and there is a set $N=\{1,\dots,n\}$ of buyers
interested in buying some of the items.
For every subset $S\subseteq E$ of items, 
player $i$ experiences value $w_i(S)\in \R$
giving rise to a \emph{valuation function}
$w_i:2^m\rightarrow \R ,i\in N,$
where $2^m$ represents the set of all subsets of $E$.
The market manager wants to determine
a price vector $\bm\lambda\in \R_+^m$ so that
all items are sold to the players and every 
player demands a subset $S_i\subseteq E$ maximizing her quasi-linear utility:
$S_i\in \arg\max_{S\subseteq E}\{ w_i(S)-\sum_{j\in S} \lambda_j \}$.
This is known as a \emph{Walrasian competitive equilibrium}.

This class of games also belongs
to the class $G^{\max}(\vec u)$ augmented with prices introduced in Assumption~\ref{ass:aggregate},
because the valuation function of a buyer only depends on her own
assigned bundle of items. 
If $X_i, i\in N$ represents the set of incidence vectors of subsets of $E$,
we can set $\vec u=(1,\dots,1)^\intercal\in \R^m$
and any pair $\vec x\in X, \bm \lambda\in \R^m_+$ 
that \emph{weakly enforces $\vec u$ with market prices} corresponds to a competitive
equilibrium. 
\iffalse
Several further variants are known in the literature
according to whether items are divisible or not and
if allocations of items need to satisfy further combinatorial
or algebraical constraints.
As we will show later, the introduced framework allows several generalizations, such as letting the valuation function
also depend on the aggregated vector of allocations of other players.
\fi
\end{example}
A related application are so-called \emph{trading networks}
as introduced by Hatfield et al.~\cite{HatfieldKNOW15}.

\begin{example}[Bilateral Trading Networks]
A bilateral trading network is represented by a directed multigraph $G = (N, E)$, where $N$ is the set of vertices and $E=\{e_1,\dots,e_m\}$ the set of edges. Vertices of the graph correspond to players and edges
represent possible bilateral trades that can take place between the  pair of incident vertices. 
For such trade $e=(s,b)\in E$, the source vertex $s$ corresponds to the seller and the sink vertex $b$ corresponds to the buyer. 
For a set of edge prices $\lambda_e\geq 0, e\in E$, we can associate
with each possible trade $e=(s,b)\in E$ 
a price $\lambda_e\geq 0$ with the understanding
that the buyer $b$ pays $\lambda_e$ to the seller $s$.
An outcome of the market is a set of 
\emph{realized trades} $S\subseteq E$ and a vector of prices $\bm\lambda\in \R_+^m$.
Given an outcome, the quasi-linear utility of a player $i \in N$  is defined as
the  
sum of the utility gained from trades plus the income
minus the cost of trades, respectively.
The utility of
realized trades is given by a function 
$\bar w_i:2^{\delta(i)}\rightarrow \R.$
As in market equilibria, the goal is to identify a subset of edges
and a price vector so that every player
gets a utility maximizing subset of trades.
The main difference to market equilibria arises as players
can act simultaneously as both, sellers and buyers in the market.
As we will show in Section~\ref{sec:trading}, this class of games also belongs
to the class $G^{\max}(\vec u)$ augmented with prices introduced in Assumption~\ref{ass:aggregate}.
\iffalse
because the valuation function of a buyer only depends on her own
assigned bundle of trades. 

If $X_i, i\in N$ represents the set of incidence vectors of subsets of $E$,
where $\{-1,1\}$ entries are used to distinguish
between the role as seller and buyer, 
we can set $\vec u=(0,\dots,0)^\intercal\in \R^m$
and any pair $\vec x\in X, \bm \lambda\in \R^m_+$ 
that \emph{enforces $\vec u$} corresponds to a competitive equilibrium, details are in Section~\ref{sec:trading}. 
\fi
\end{example}

The next application resides in the area of congestion control in communication networks.
\begin{example}[Congestion Control in Communication Networks]
We consider a  model
of Kelly et al.~\cite{Kelly98}  in the domain of TCP-based congestion control. We are given a directed or undirected \emph{capacitated} graph $G=(V,E,\vec c)$,
where $V$ are the nodes, $E$ the edge set with $|E|=m$ and
$\vec c \in \R_+^m$ denotes the edge capacities.
There is a set of players $N= \{1, \dots,
n\}$ and each $i\in N$ is associated with an end-to-end pair $(s_i,t_i)\in V\times V$ and a non-decreasing and concave bandwidth utility function $U_i:\R_+\rightarrow\R_+$
measuring the received benefit from sending net flow from $s_i$ to $t_i$.
The strategy space $X_i$ of a player represents a flow polyhedron 
and for a flow $\vec x_i\in X_i$ with value $\val(\vec x_i)$ the
received bandwidth utility is equal to $U_i( \val(\vec x_i))$.
The goal in this setting is to determine a price vector $\bm \lambda\in \R_+^m$ so that
a strategy distribution $\vec x^*$ is induced as an equilibrium
respecting the network capacities $\vec c$ and, hence, avoids
congestion. The equilibrium condition is given by
\[ \vec x_i^*\in \arg\max\{U_i( \val(\vec x_i))-\bm\lambda^\intercal\vec x_i \vert \vec x_i\in X_i\} \text{ for all $i\in N$.}\]
This model fits to the class $G^{\max}(\vec u)$ augmented with prices:
The utility function of a player only depends on the own action
and with $\vec u:=\vec c$ we obtain the desired structure.
\iffalse
For maximization problems over network flows, it is clear that in a capacitated graph only bottleneck edges are saturated, 
thus, the goal is to determine a price vector $\bm \lambda\in \R_+^m$
so that the network capacities are \emph{weakly enforced with market prices}.
\fi
\end{example}

\subsection{Overview of Results, Main Techniques and Organization of the Paper}
In Section~\ref{sec:model} and~\ref{sec:pricing}, we introduced a resource allocation
model and motivated the question of enforceability
of load vectors induced by equilibrium profiles
with respect to anonymous resource prices.
In Section~\ref{sec:gap}, 
we will revisit a well-known duality-based framework
and we prove a  complete characterization of enforceability:
\begin{framed}
Theorem~\ref{thm:main}: $\vec u\in \R^m$ is enforceable by  $(\vec x^*,\bm\lambda)$
if and only if $(\vec x^*,\bm\lambda)$ yields 
zero duality gap for the master problem 
$\min\{\sum_{i\in N} \pi_i(\vec u,\vec x_i)\vert \ell(x)\leq \vec u, x\in X\}$ and $x^*$
satisfies  $\ell(\vec x^*)=\vec u$.
\end{framed}
The  above result only requires
the separability structure of utility/cost functions and the conditions on the duality gap.\footnote{The if-direction is well-known, see ~\cite{BertsekasGallager1992,Palomar:2006,ScutariPFP10}. The only-if direction roughly corresponds to the first
welfare theorem in economics saying that every pricing equilibrium maximizes social welfare, however, the theorem asks for a slightly stronger condition  (strong duality) and the fact that an inequality must be tight. }
Otherwise, the  strategy sets and utility/cost functions 
are not restricted, for instance, they may be non-convex.
However,  checking whether
or not a non-convex master problem exhibits zero duality gap
and admits an optimal solution with the desired property may be very difficult to prove. In this regard, for any $G^{\min}(\vec u)$,
 we introduce in Section~\ref{sec:convex-relaxations} the notion of a \emph{convex
relaxation} $G^{\min-\conv}(\vec u)$ using the concept of \emph{convex subfunctionals}. 
We derive the following characterization:
\begin{framed}
 Theorem~\ref{thm:main-convexification}: $\vec u$ is enforceable for $G^{\min}(\vec u)$ via $(\vec x^*,\lambda)$ if and only if $(\vec x^*,\lambda)$ enforces $\vec u$ for the convexified game $G^{\min-\conv}(\vec u)$.
 \end{framed}
While the new game $G^{\min-\conv}(\vec u)$  is now convex
and, thus, more accessible, the complexity of the original non-convex
problem is in some sense shifted to the representation of the convex subfunctionals.
For two  game classes, however, namely (1) games with finite strategy spaces
or (2) games for which the convex hull of the strategy space is finite
and the cost/utility functions is concave/convex, we can show that
the convex subfunctionals are representable as optimal
solutions of an underlying LP. This leads to the third main result: 
\begin{framed}
Theorem~\ref{thm:convex-hull-finite} and Theorem~\ref{concave-main}: For (1) and (2):  $\vec u$ is enforceable for $G^{\min}(\vec u)$
if and only if  the master LP for $G^{\min-\conv}(\vec u)$
admits optimal solutions that are feasible for $G^{\min}(\vec u)$ .
\end{framed}
The dual of the master LP can be solved  in polynomial time via the ellipsoid
method, if there is an efficient separation oracle.
The separation oracle for the two problem classes reduces to the so-called \emph{demand problem}, where for every player one is given prices and
the problem is to compute an optimal strategy. The complexity of the demand- versus
the master problem~\ref{price-opt} can then be used to establish impossibility
results for enforceability using complexity-theoretic
assumption (like $P\neq NP$). The connection between the complexity of the
demand problem and that of the master problem (or welfare maximization problem)  has been discovered first by
Roughgarden and Talgam-Cohen~\cite{Roughgarden:2015}
in the context of pricing problems for Walrasian market
equilibria.

We
then consider in Section~\ref{sec:integral} the case of \emph{integrality} of strategy spaces.
It follows that for master problems admitting a fractional relaxation 
with zero duality gap and integer optimal solutions, the sufficiency condition of Theorem~\ref{thm:main} is satisfied.
 For polyhedral integral strategy spaces, the powerful methods from polyhedral combinatorics can be used to categorize
 cases for which such relaxations exist.
 In this regard, we show 
 two prototypical results:
 \begin{framed}
 \begin{enumerate} 
 \item  Theorem~\ref{thm:box-integral}
 gives an enforceability result for games with homogenous additive linear utilities/costs and (box) totally-dual-integral and decomposable
aggregation polytopes.
\item  Theorem~\ref{thm:polymatroid-main}  
gives an enforceability result for games with player-specific additive linear utilities/costs on \emph{polymatroidal} strategy spaces.
\end{enumerate}
\end{framed}

In Section~\ref{sec:monotone}, we turn to models 
for which the private cost/utility function
intrinsically depends on the aggregated load vector $\ell(\vec x)$,
that is, it has the form $\pi_i(\ell(\vec x),\vec x_i), i\in N$
and is \emph{not} separable anymore.
A prime example is an atomic congestion game, where
any change of strategies has an effect on the perceived
cost since the load vector changes.
We identify an expressive class of games
that we term \emph{monotone aggregative games}
which include among others congestion games with nondecreasing cost functions. 
\begin{framed}
Theorem~\ref{thm:mag}: Let $G^{\min-\moag}$ be a monotone aggregative game.
If  $\vec u$ is enforceable for
the game $G^{\min}(\vec u)$ with 
$\pi_i(\ell(\vec x),\vec x_i):=\pi_i(\vec u,\vec x_i), i\in N, x\in X$,
then $\vec u$ is also enforceable for $G^{\min-\moag}$.
\end{framed}
This way, we can translate several positive enforceability results
to monotone aggregative game including atomic congestion games.
With these results and methods at hand, we apply the framework 
to the four application domains. 
\paragraph{Tolls in traffic networks.}
In Section~\ref{sec:congestion-games}, we consider the problem of
defining tolls in order to enforce certain load vectors as
Wardrop equilibrium. For nonatomic network games, we reprove and generalize in Corollary~\ref{cor:nonatomic} a characterization
of enforceable load vectors by 
Yang and Huang~\cite{Yang04}, Fleischer et al.~\cite{Fleischer04} and Karakostas and Kolliopoulos~\cite{Karakostas04}.

Then we turn to atomic congestion games. For general nondecreasing homogeneous cost functions, we show that polytopal congestion games can be analyzed
using Theorem~\ref{thm:box-integral}. It turns out that for a wide
classes of congestion games (matroid games, single-source network games,
$r$-arborescences, matching games, and more) the
defining aggregation polytope is box-integral and decomposable leading to existence
results of enforcing tolls (Corollary~\ref{cor:consequences}).
For all these settings, a congestion vector $\vec u$
minimizing the social cost  can be computed in polynomial
time (see Del Pia et al.~\cite{PiaFM17} and Kleer and Sch\"afer~\cite{KleerS17}) and the space of enforcing
prices can be described by a compact linear formulation.
It follows that for a fixed enforceable capacity vector $\vec u$, arbitrary linear objective functions 
(like maximum or minimum revenue) can 
be efficiently optimized over the price/allocation space.
Besides single-source network games (see Fotakis and Spirakis~\cite{FotakisS08} and
Fotakis et al.~\cite{FotakisKK10}), these results
were not known before.

Then, we study the more challenging case of
atomic congestion games with nondecreasing player-specific 
cost functions on the resources. We prove -- using 
 Theorem~\ref{thm:polymatroid-main} -- that for polymatroidal
strategy spaces, one can obtain existence results (Corollary~\ref{cor:congestion-polymatroid}). 
To the best of our knowledge, these are the first existence results of tolls for congestion games with
player-specific cost functions. 
\iffalse
We complement this result
by showing in Theorem~\ref{thm:hard-hetero} that already the case of homogenous cost functions with heterogeneous players
is considerably harder than the purely homogenous case: even for
symmetric $s$,$t$ network games the corresponding (compact) LP-relaxation
cannot be integral (unless $P=NP$).
 \fi
\paragraph{Market equilibria.}

In Section~\ref{sec:market-equilibria}, we study (indivisible) single-, multi-item, or package auctions  and show that the existence of Walrasian equilibria can be studied within the framework.
Using the fact that in all these models, the 
quasi-linear utility function is separable and the strategy space
of every player (buyer or seller) consists of a finite
point set we can use the relationship
between a game and its convexified game. This way,  we can reprove classical LP-characterization results
of the existence of Walrasian equilibria
by Bikchandani and Mamer~\cite{Bikchandani1997}, Bikchandani
and Ostroy~\cite{BikhchandaniO02} as well as more recent LP characterizations by
Candogan et al.~\cite{Candogan2018} and  
Roughgarden and Talgam-Cohen~\cite{Roughgarden:2015}.
For gross-substitute valuations we reprove the existence
of Walrasian equilibria using methods of discrete
convexity and $M$-convexity (see Murota~\cite{Murota:2003}).

Then we consider a class of valuations
for multi-unit items that we term \emph{separable additive} valuations
with \emph{negative externalities}.
The idea is that items are of different type but
may be sold at a certain multiplicity and the values for received items are additive.
The precise item values may depend on the allocation
vector. This dependency is assumed to model
negative externalities, that is, - roughly speaking -  if an item type
is sold to more players, the value goes down.
For this class of valuations, we prove in Corollary~\ref{cor:market-polymatroid} that for general polymatroidal environments,
Walrasian equilibria exist.  
\paragraph{Trading networks.}
In Section~\ref{sec:trading}, we study trading networks
as introduced by Hatfield et al.~\cite{HatfieldKNOW15}.
We  show that also this class of games fits
into the framework. The main
conceptual difference to the previous market
equilibrium setting is that players
may be both buyers and sellers at the same time.
As our model allows negative resource consumption, we just use
 $\{-1,1\}$ entries in an allocation vector to distinguish buy
 or sell activities.
Using the relationship
between a game and its convexified game, we prove an LP-characterization result
of the existence of trading equilibria (Corollary~\ref{cor:trading-LP}).
This characterization was, to the best of our knowledge, not known before.
For gross-substitute valuations, we give a simple proof for the existence
of Walrasian equilibria using again $M$-convexity arguments.
%This proof is considerably simpler than that of Hatfield et %al.~\cite{HatfieldKNOW15}.

\paragraph{Congestion control in communication networks.}
In the final Section~\ref{sec:congestion-control},
we consider congestion control problems in communication
networks using a flow-based model proposed by Kelly et al.~\cite{Kelly98}.
We first reprove an existence result of enforceable capacity vectors
of Kelly et al.~\cite{Kelly98}. Then, we turn to the much less
explored model of \emph{integral flows},
where a discrete unit-packet size is given.  With the previous
results related to TDI systems, we prove that for single-source
networks with identical linearly increasing bandwidth utility functions, every nonnegative
capacity vector is weakly enforceable with market prices
(Corollary~\ref{cor:kelly-integral}).
We complement this result by showing in Proposition~\ref{prop:congestion-control-reduction} that already for two-player
instances with different source-sink pairs and linear and identical  (capped)
bandwidth utility functions, not every $\vec u$ is weakly enforceable by market prices, unless $P=NP$. For this result, we use the LP-characterization of Theorem~\ref{concave-main} and then show that the demand
problem is polynomial time solvable while the master-problem is NP-hard.

\subsection{Related Work}
As outlined in the introduction, the topic of pricing resources 
concerns different streams of literature and it seems
impossible to give a complete overview here.
Lagrangian multipliers date back to the 18th century and their use in terms of \emph{shadow prices}
measuring the change of the optimal value function
for marginal changes of the right-hand sides of constraints is well-known -- 
assuming some constraint qualification conditions, see 
for instance Boyd and Vandenberghe~\cite{Boyd:2004}.

Our first main result  (Theorem~\ref{thm:main}) relies on a decomposition property of the Lagrangian
(for separable problems) and the use of Lagrange multipliers for pricing the resources. This approach is by no means new and has
been developed in several facets before, see for instance Dantzig and Wolfe~\cite{Dantzig:1960} and Bertsekas and Ghallager~\cite{BertsekasGallager1992}.
Dantzig and Wolfe~\cite{Dantzig:1960} used this principle
for their  celebrated decomposition framework for solving certain
linear (integer) programming problems.
Bertsekas and Ghallager~\cite{BertsekasGallager1992}, Palomar and Chiang~\cite{Palomar:2006} and Scutari et al.~\cite{ScutariPFP10} described how the Lagrangian of a general separable optimization problem
\[ \max\Big\{\sum_{i\in N}U_i(\vec x_i) \vert \vec x_i\in X_i, i\in N, \; \sum_{i\in N}h_i(\vec x_i)\leq \vec u\Big\}\]
can be decomposed into $n$ independent problems.
These works  describe the close connection
between strong duality and the existence of enforcing dual prices.
One subtle difference of this model
to ours is the parameterization
of the cost/utility functions  $\pi_i(\vec u,\vec x_i) $ with respect to the capacity vector $\vec u$.
This degree of freedom allows to model
dependencies of targeted capacity vectors with respect
to the intrinsic cost/utilities - a prime example appears in nonatomic congestion games, where the cost function of an agent \emph{only} depends on the aggregated load vector. Moreover, this dependency allows
to model \emph{externalities} with respect to allocations
which are not directly possible in the previous formulations.
In contrast to most works in the ``dual-decomposition'' area, we systematically investigate the impact of
\emph{non-convexities} of the cost/utility functions
and the strategy spaces (e.g., integrality of strategies)
on the resulting enforceability properties.

\paragraph{Convexification of Non-Convex Models.}
The idea of convexifying a non-convex economic model
dates back to the late sixties starting with the work 
of Shapley and Shubik~\cite{Shapley66} and Starr~\cite{Starr1969}.
Starr~\cite{Starr1969} considered a standard Arrow-Debreau exchange economy 
without convexity assumptions on production or consumption
sets nor on the preference ordering. The analysis of the existence of competitive market equilibria is based
on a \emph{convexified economy} in which the
convex hull of production or consumption sets and  the convex hull of the epigraph with respect to the preference orderings are considered (see also later related works of Henry~\cite{henry72}, Moore et al.~\cite{Moore72}
and Svensson~\cite{svensson84}). 
By separation arguments, this convexified economy permits
a competitive equilibrium (called a synthetic convex equilibrium).
A quasi-equilibrium lives in the original non-convex model and is defined as  a closest approximation within  w.r.t the synthetic equilibrium. 
With the Shapley-Folkman Theorem (which appeared inside the paper
of Starr) the approximation guarantee can be parameterized in terms
of the number of commodities or number of traders involved.\footnote{The
bound was recently improved by Budish and Reny~\cite{Budish20}.}
For large markets (number of traders tends to be large) this distance
vanishes.\footnote{In the spirit of large markets, 
Aumann~\cite{Aumann66} derived a very general existence result
of competitive equilibria
assuming a continuum of traders but without any convexity assumptions.}
The approach of convexifying a game in this paper is qualitatively 
similar to that of Starr. The main difference lies in the representation
of the convexified game. Instead of convexifying the epigraph
of utility level sets as in Starr, we explicitly use convex envelopes
of the utility functions which allow (in the context of separable problems) to define a convex master-problem. This way, we obtain for 
(1) games with finite strategy spaces
or (2) games for which the convex hull of the strategy space is finitely generated
and the cost/utility functions is concave/convex
a tractable LP formulation for the master problem.
The assumption (1) for instance is fulfilled  for exchange markets
with indivisible items and the representation of the convexified game
corresponds to the so-called configuration LP of Bikchandani
and Mamer~\cite{Bikchandani1997} and  Bikchandani
and Ostroy~\cite{BikhchandaniO02}.

\paragraph{Tolls in Traffic Networks.}
A large body of work in the area of transportation networks is
concerned with congestion toll pricing.
Beckmann et al.~\cite{Beckmann56} showed that for the Wardrop model
with homogeneous users, charging the difference between the marginal
cost and the real cost in the socially optimal solution (marginal cost
pricing) leads to an equilibrium flow which is optimal. Cole et al.~\cite{cole2003pne}
considered the case of heterogeneous users,
that is, users value latency relative to monetary cost differently.
For single-commodity networks, the authors showed the existence of
tolls that induce an optimal flow as Nash flow. Yang and Huang~\cite{Yang04}, Fleischer et al.~\cite{Fleischer04} and Karakostas and Kolliopoulos~\cite{Karakostas04} proved that there are tolls inducing an optimal flow for heterogenous
users even in general networks - all proofs are based on linear programming duality.
Swamy~\cite{swamy07} and Yang and Zhang~\cite{YangX08} proved
the existence of optimal tolls for the
atomic splittable model using convex programming duality. 

For atomic (unsplittable) network congestion games much less
is known regarding the existence of tolls.
Caragiannis et al.~\cite{CaragiannisKK10} studied the existence of tolls for singleton congestion games.
%For results related to the current model, they showed
%that for singleton strategies and  nondecreasing homogeneous cost
%functions, the minimum cost solution can be enforced by tolls.
Fotakis and Spirakis~\cite{FotakisS08} proved the existence
of tolls inducing any acyclic integral flow for symmetric $s$,$t$ network games with homogeneous players.
Fotakis et al.~\cite{FotakisKK10} further extended this result to heterogeneous players and networks with a common source but different sinks. Marden et al.~\cite{Marden09} transferred the idea of charging marginal cost tolls to congestion games and showed
the existence of tolls enforcing the load vector of a socially
optimal strategy distribution. Very recently, Chandan et al.~\cite{ChandanPFM19} derived an optimization
formulation computing optimal tolls
minimizing the resulting price of anarchy.

\paragraph{Market Equilibria.}
For the problem of allocating indivisible single-unit items,
there are several characterizations of the
existence of competitive equilibria related to
the gross-substitute property of valuations, see 
Kelso and Crawford~\cite{Kelso82}, Gul and Stachetti~\cite{Gul99}
and Ausubel and Milgrom~\cite{Ausubel2002}.
Several works established connections
of the equilibrium existence problem w.r.t. LP-duality and integrality (see Bikchandani
and Mamer~\cite{Bikchandani1997}, Bikchandani
and Ostroy~\cite{BikhchandaniO02} and Shapley and Shubik~\cite{Shapley1971}). Murota and Tamura~\cite{MurotaT03,Murota:2003}
established connections between the gross substitutability property
and M-convexity properties of demand sets and valuations.
Yokote~\cite{Yokote2018} recently proved that the existence
of Walrasian equilibria follows from a duality property in
discrete convexity.

For multi-unit items, several recent papers
studied the existence of Walrasian equilibria. 
Danilov et al.~\cite{Danilov2001} investigated the existence
of Walrasian equilibria in multi-unit auctions and identified general
conditions on the demand sets and valuations
related to discrete convexity, see also
Milgrom and Struluvici~\cite{Milgrom2009} and Ausubel~\cite{Ausubel2006}. 
Baldwin and Klemperer~\cite{Baldwin2019} explored a connection with tropical geometry and gave necessary and sufficient conditions for the existence of competitive equilibrium in product-mix auctions of indivisible goods, see also Sun and Yang~\cite{Sun2009}.
% and
%Tran and Yu~\cite{Tran2015}.
%gave a new proof of the sufficiency condition of~\cite{Baldwin2019} using a %unimodularity theorem in integer programming.
For a comparison of the above works especially
with respect to the role of discrete convexity, we refer to the excellent survey
of Shioura and Tamura~\cite{Shioura2015}. Candogan et al.~\cite{Candogan2018,CandoganP18} showed that valuations classes
(beyond GS valuations) based on graphical structures also imply the existence of Walrasian equilibria. Their proof also uses
integrality of optimal solutions of an associated linear min-cost
flow formulation and linear programming formulation, respectively.

Our existence result for polymatroid environments
 differs to these
previous works in the sense
that we allow valuations to depend on the allocation
of items to other players (negative externalities).
Much fewer works allow for externalities in valuation
functions, see for instance Zame and Noguchi~\cite{Zame2006}.
Models with positive (network-based) externalities  have been considered by
Candogan et al.~\cite{Candogan2012}. Bhattacharya et al.~\cite{BhattacharyaKMX11} considered a setting
with weighted negative network-based externalities and unit-demand buyers.
Bikchandani et al.~\cite{Bikchandani2011} consider a problem of
selling a base of polymatroid. In their model, however, the prices
are not anonymous (rather VCG) for several items of the same type. The same holds true for
Goel et al.~\cite{Goel2015}  who also consider polymatroids
even  with budget constraints.
Feldman et al.~\cite{FeldmanGL16} proposed the notion
of combinatorial Walrasian equilibria, where items
can be packed a priori into bundles. This
ensures the existence of equilibria with approximately optimal welfare
guarantees.
\iffalse
As mentioned before, Roughgarden and Talgam-Cohen~\cite{Roughgarden:2015}
established an interesting and far-reaching connection between the equilibrium existence of Walrasian equilibria and
the computational complexity of the allocation and demand problems.
With our LP-characterization given in Theorem~\ref{thm:convex-hull-finite} and Theorem~\ref{concave-main}, we can apply the methodology of
Roughgarden and Talgam-Cohen even to a wider class of problems.
\fi
%, where
%in the former, one computes a welfare or revenue optimal
%allocation given the valuations.

\paragraph{Trading Networks.}

Hatfield et al.~\cite{Hatfield13,HatfieldKNOW15}
introduced the model of trading networks
and established existence and characterization results
for so-called fully-substitutable valuations -- a generalization
of gross-substitutable valuations.
Ikebe et al.~\cite{Ikebe15}  generalized the model
of Hatfield et al. by using
certain discrete concave utility functions for which they derived existence results.
Subsequently, Candogan et al.~\cite{Candogan:2016}
reduced the problem of computing competitive equilibria
to a submodular  flow problem on a suitably defined
network. This way, they established the polynomial time computation of 
market equilibria for fully substitutable valuations.
Further generalizations regarding the inclusion of taxes and
other monetary transfers appear in Fleiner et al.~\cite{Fleiner19}.

\paragraph{Congestion Control.}
 Kelly et al.~\cite{Kelly98} proposed to model congestion control
 via analyzing optimal solutions of a
convex optimization
problem, where an aggregated bandwidth utility subject to network
capacity constraints is maximized.
By dualizing the problem and then decomposing terms
(as we do in this paper),
it is shown that Lagrangian multipliers correspond to equilibrium
enforcing congestion prices. For an overview
on more related work in this area, we refer to the book by 
Srikant~\cite{srikant03}.
Kelly and Vazirani~\cite{Vazirani2002} drew connections between
market equilibrium computation and the congestion control model
of Kelly.
Cominetti et al.~\cite{Cominetti:2014}
also studied the convex programming formulation of Kelly et al. and
established  connections to the Wardrop equilibrium model.
The most obvious difference of these work to ours is
that they assume convex strategy spaces and concave utility
functions. Our framework allows to add integrality conditions
or non-convexities to the model.
\section{Connection to Lagrangean Duality in Optimization}\label{sec:gap}
In the following, we distinguish between \emph{cost minimization problems}
and \emph{utility maximization problems}.
We explicitly prove our main results in the realm
of  cost minimization but all arguments
carry directly over to the maximization case.
For later referral, we summarize the 
results for the maximization case at the end of the section.

%\subsection{Cost Minimization Problems}\label{subsec:cost-min}
For a game $G^{\min}(\vec u)$, we define the following minimization problem
that we call \emph{master problem}:
\begin{framed}
\begin{align}\tag{$P^{\min}(\vec u)$}\label{price-opt}
\min\; & \pi(\vec x) \\
\text{s.t.: } &  \ell_j(\vec x) \leq u_j, \; j\in E, \label{eq:inequality}\\
\vec x_i & \in X_i, \; i=1,\dots,n,\notag
\end{align}
where the objective function is defined as
$ \pi(\vec x):=\sum_{i\in N}\pi_i(\vec u,\vec x_i).$
\end{framed}
We  assume in the formulation of~\ref{price-opt} that a global minimum actually
exists. 
The Lagrangian function for problem~\ref{price-opt} becomes
$ L(\vec x,\bm\lambda):=\pi(\vec x) +\bm\lambda^\intercal (\ell(\vec x)-\vec u) ,\; \bm \lambda\in \R_+^m,$
and we can define the Lagrangian-dual as:
$ \mu : \R_+^m \rightarrow\R, \;\;
\mu(\bm \lambda)=\inf_{\vec x \in X} L(\vec x,\bm \lambda)=\inf_{ \vec x \in X}\{\pi(\vec x)+\bm\lambda^\intercal (\ell(\vec x)-\vec u)\}.
$
We assume that $\mu(\bm \lambda)=-\infty$, if $L(\vec x,\bm \lambda)$ 
is not bounded from below on $X$.
The \emph{dual problem} is defined as:
\begin{align}\label{price-dual-min}
\tag{$D^{\min}(\vec u)$} \sup_{\bm \lambda\geq 0} \mu(\bm \lambda)\end{align}

\begin{definition}
Problem~\ref{price-opt} has zero-duality gap,
if there is $\bm \lambda^*\in \R_+^m$ and $\vec x^*\in X$ with
$\pi(\vec x^*)=\mu(\bm \lambda^*).$
In this case, we say that the pair $(\vec x^*,\bm \lambda^*)$ is
primal-dual optimal.
\end{definition}

If problem~\ref{price-opt} has zero-duality gap,
the two solutions $\bm \lambda^*\in \R_+^m$ and $\vec x^*\in X$
are optimal for their respective problems~\ref{price-dual-min} and \ref{price-opt}
and  infima/suprema in the definition of $\mu$ become a minimum/maximum.

We now show a key structure, namely that 
the Lagrangian dual can be decomposed into $n$ independent subproblems.
This decomposition step is classical for separable optimization problems,
see Bertsekas and Ghallager~\cite{BertsekasGallager1992}.
\begin{lemma}\label{eq:decomposition}
Let $\bm \lambda\in \R_+^m$. For a problem of type~\ref{price-opt}, the following holds true:
\begin{equation}
\vec x^*\in\arg\min_{\vec x \in X} L(\vec x,\bm \lambda)\Leftrightarrow \vec x_i^*\in\arg\min_{\vec x_i\in X_i} \{\pi_i(\vec u,\vec x_i)+\bm\lambda^\intercal g_i(\vec x_{i}) \} \text{ for all }i\in N.
\end{equation}
\end{lemma}
\begin{proof} 
We calculate:
\begin{align*}
%\mu(\bm \lambda)&=
\min_{\vec x \in X} L(\vec x,\bm \lambda)&=\min_{\vec x_i\in X_i,i\in N}\Big\{\sum_{i\in N}(\pi_i(\vec u,\vec x_i)+\sum_{j=1}^m \lambda_j( g_{ij}(\vec x_{i})-u_j)) \Big\}\\
&=\sum_{i\in N}  \min_{\vec x_i\in X_i} \Big\{\pi_i(\vec u, \vec x_i)+\sum_{j=1}^m \lambda_j ( g_{ij}(\vec x_{i})-u_j) \Big\},
\end{align*}
where the first equality follows by the linearity of $\ell(\vec x)$
w.r.t. $g_i, i\in N$
and the last equality by the assumption that $\pi_i(\vec u, \vec x_i)$
only depends on $\vec x_i\in X_i$. Because taking the minimum
is independent of the constant $-\sum_{j=1}^m\lambda_j u_j$, the lemma follows.
\end{proof}
We obtain the following result.
\begin{theorem}\label{thm:main}
The following equivalences hold for $G^{\min}(\vec u)$.
\begin{enumerate}
\item\label{enum:main1} A capacity vector $\vec u\in \R^m$ is enforceable
via $(\vec x^*,\bm \lambda^*)$ if and only if $(\vec x^*,\bm \lambda^*)$  has zero duality gap for~\ref{price-opt}  and $\vec x^*$ 
satisfies~\eqref{eq:inequality} with equality. 
\item\label{enum:main2} 
A capacity vector $\vec u\in \R^m$ is weakly enforceable
via $(\vec x^*,\bm \lambda^*)$ with market clearing prices $\bm \lambda^*$ if and only if $(\vec x^*,\bm \lambda^*)$ has zero duality gap for~\ref{price-opt}. 
\item\label{enum:main-unique} 
A capacity vector $\vec u\in \R^m$ is uniquely enforceable
via $(\vec x^*,\bm \lambda^*)$ if and only if $(\vec x^*,\bm \lambda^*)$ has zero duality gap for~\ref{price-opt}  and $\vec x^*$ is
a unique optimal solution for~\ref{price-opt}
satisfying~\eqref{eq:inequality} with equality.  
\end{enumerate}
\end{theorem}
\begin{proof}
For the proof it suffices to show~\ref{enum:main2}.,
since~\ref{enum:main1}. satisfies all conditions
of~\ref{enum:main2}. except that $\ell(\vec x^*)=\vec u$
holds true for either side of the equivalence in ~\ref{enum:main1}.
Statement~\ref{enum:main-unique}. follows directly from~\ref{enum:main1}.
as on both sides of ~\ref{enum:main-unique}. uniqueness of $\vec x^*$ is assumed.\\
%The statement~\ref{enum:main3}. follows from the if-condition
%of equivalence~\ref{enum:main2}.\\
For \ref{enum:main2}.: $\Leftarrow:$
Assume there are $\bm\lambda^*\in \R_+^m, \vec x^*\in X$ with  $\ell(\vec x^*)\leq \vec u$
so that $\mu(\bm\lambda^*)=\pi(\vec x^*)$. We obtain
\[ \mu(\bm\lambda^*)=\min_{\vec x \in X} \{\pi(\vec x)+(\bm\lambda^*)^\intercal(\ell(\vec x) -\vec u)\}\leq \pi(\vec x^*)+(\bm\lambda^*)^\intercal(\ell(\vec x^*) -\vec u)\leq \pi(\vec x^*)=\mu(\bm\lambda^*).\]
Hence, all inequalities must be tight leading to
$(\bm\lambda^*)^\intercal(\ell(\vec x^*) -\vec u)=0$  as claimed.
 It remains to prove Condition~\ref{cond2}.
With  $\vec x^*\in \arg\min_{\vec x \in X} L(\vec x,\bm\lambda^*)$ we get
\begin{align}\notag \vec x^*\in\arg\min_{\vec x \in X} L(\vec x,\bm\lambda^*)\underset{Lem.~\ref{eq:decomposition}}\Leftrightarrow& \vec x_i^*\in\arg\min_{\vec x_i\in X_i} \{\pi_i(\vec u,\vec x_i)+\sum_{j\in E} \lambda_j^*  g_{ij}(\vec x_{i}) \} \text{ for all }i\in N.
\end{align}

$\Rightarrow:$
Let $\vec u\in \R^m$ be weakly enforceable by some $\vec x^*\in X$ 
with market clearing prices $\bm \lambda^*\in \R_+^m$, that is,  $(\vec x^*,\bm \lambda^*)$ satisfy
$\ell(\vec x^*)\leq \vec u, (\lambda^*)^\intercal (\ell(\vec x^*)-\vec u)=0$ and
$ \vec x^*_i\in\arg\min_{\vec x_i\in X_i}\left\{ \pi_i(\vec u,\vec x_i)+(\bm\lambda^*)^\intercal  g_{i}(\vec x_{i})\right\} \text{ for all }i\in N.$
We calculate
\begin{align}\notag
\mu(\bm \lambda^*)&=\inf_{\vec x \in X} \{\pi(\vec x)+(\bm\lambda^*)^\intercal(\ell(\vec x) -\vec u)\}\\
\iffalse
\label{eq:4}
&=\inf_{\vec x \in X} \{\sum_{i\in N} \pi_i(\vec u, \vec x_i)+(\bm\lambda^*)^\intercal g_i(\vec x_i)\}
-(\bm\lambda^*)^\intercal\vec u\\\label{eq:5}
&=\sum_{i\in N}  \inf_{\vec x_i\in X_i} \left\{\pi_i(\vec u, \vec x_i)+(\bm\lambda^*)^\intercal g_i(\vec x_i)\right\}-(\bm\lambda^*)^\intercal\vec u\\
&=\sum_{i\in N}  \min_{\vec x_i\in X_i} \left\{\pi_i(\vec u, \vec x_i)+ \notag(\bm\lambda^*)^\intercal g_i(\vec x_i)\right\}-(\bm\lambda^*)^\intercal\vec u\\%
%&=\sum_{i\in N}  \inf_{\vec x_i\in X_i} \left\{\pi_i(\ell(\vec x), \vec x_i)+(\bm\lambda^*)^\intercal\vec x_i\right\}-(\bm\lambda^*)^\intercal\vec u\\
\fi
\label{eq:6} &= \pi(\vec x^*)+(\bm\lambda^*)^\intercal \ell(\vec x^*)-(\bm\lambda^*)^\intercal\vec u\\\label{eq:7}
&=\pi(\vec x^*),
\end{align}
\iffalse
where~\eqref{eq:4} follows from the definition of $\pi(\vec x)$ and the linearity of  $\ell(\vec x)$ w.r.t. $g_i(\vec x_i), i\in N$,~\eqref{eq:5}
follows because $\pi_i(\vec u, \vec x_i)$ only depends on $\vec x_i$,~
\fi
where
\eqref{eq:6} follows from Lemma~\ref{eq:decomposition} and~\eqref{eq:7}
uses the market price condition $(\bm \lambda^*)^\intercal (\ell(\vec x^*)-\vec u)=0$.
Hence, strong duality holds for the pair $(\vec x^*,\bm\lambda^*)$.
\end{proof}

As mentioned before, the if-direction of the above characterizations are
well known in  the literature, see, e.g.\cite{BertsekasGallager1992,Kelly98,Palomar:2006,ScutariPFP10}.
We remark here that the theorem does not rely on any assumption
on the feasible sets $X_i$ nor on the functions $\pi_i(\vec u,\vec x_i), i\in N$
as long as~\ref{price-opt} has zero duality gap.
In the optimization literature, several classes
of optimization problems are
known to have zero duality gap even \emph{without} convexity
of feasible sets and objective functions, see for instance Zheng et al.~\cite{Zheng2012}.
In cost minimization games, the feasible sets $X_i$
 usually contain some sort of covering conditions
 on the resource consumption.
 For example in network routing, one needs
 to send some prescribed amount of flow.
 In this regard, we introduce a natural candidate set of vectors $\vec u$
 for which we know that any feasible solution satisfying~\eqref{eq:inequality} does so with equality.
\begin{definition}\label{def:minimal}
A vector $\vec u\in \R^m$ is called \emph{minimal} for $X$,
if there are strictly increasing functions $h_j:\R\rightarrow\R ,j\in E$ such that
$ \vec u\in\arg\min_{\vec u'\in \R^m} \left\{\sum_{j\in E}h_j(u'_j)\middle\vert \;\;\exists \vec x\in X \text{ with }\ell(\vec x)\leq \vec u'\right\}.$
\end{definition}
The above definition has been previously used 
by Fleischer et al.~\cite{Fleischer04} in the context
of enforcing tolls in nonatomic congestion games.
\begin{corollary}\label{thm:main2}
%Let $X_i, i\in N$ be non-empty convex sets
%and assume that $\pi_i, i\in N$ are convex functions.
Let $\vec u\in \R^m$ be minimal for $X$.
Then, the following two statements are equivalent:
\begin{enumerate}
\item\label{char1} $\vec u$ is enforceable via price vector $\bm\lambda^*\in\R_+^m$
and  $\vec x^*\in X$.
\item\label{char2}   $(\vec x^*,\bm \lambda^*)$
satisfies $\pi(\vec x^*)=\mu(\bm \lambda^*)$.
\end{enumerate}
\end{corollary}
The only difference to Theorem~\ref{thm:main}
is  that by minimality of $\vec u$,
we  get $\ell(\vec x)=\vec u$
for any feasible solution of~\ref{price-opt},
therefore, tightness of inequality~\eqref{eq:inequality}
is already satisfied.

Let us now consider the important special
case of \emph{convex} optimization problems.
\begin{corollary}\label{cor:convex}
Let $X_i, i\in N$ be nonempty convex sets 
and assume that $\pi_i(\vec u, \vec x_i), g_i(\vec x_i), i\in N$ are convex functions
over $X_i$.
Let $\vec u\in \R^m$ be minimal and suppose there exists
$\vec x^0\in \relint\left(\{\vec x\in X\vert \ell(\vec x)\leq \vec u\}\right) $,
where $\relint(U)$ denotes the relative topological interior of $U\subset \R^{nm}$.
Then, $\vec u$ is enforceable.  
If $\pi(\vec x)=\sum_{i\in N}\pi_i(\vec u, \vec x_i)$ is strictly convex 
over $X$, then $\vec u$ is uniquely enforceable.\end{corollary}
\iffalse
\begin{proof}
For~\ref{price-opt}, we have 
a convex objective over non-empty convex set.
Since~\ref{price-opt} is feasible and Slater's constrained qualification condition (cf.~\cite{slater1959}) is satisfied, we get that~\ref{price-opt}  has zero duality gap
and the result follows from Corollary~\ref{thm:main2}.

Unique enforceability follows  from uniqueness of the optimal
solution of~\ref{price-opt}.
\end{proof}
\fi
We moved analogous results for maximization problems to the appendix~\ref{subsec:max}.

\section{Convexified Games}\label{sec:convex-relaxations}
So far, the strategy spaces $X_i, i\in N$ and the cost functions $\pi_i,i\in N$
of a game $G^{\min}(\vec u)$ were 
not restricted and are allowed to be non-convex.
For instance integrality restrictions in $X_i\subset\Z^m, i\in N$ are 
allowed. In what follows, 
we connect $G^{\min}(\vec u)$
with a related convexified game $G^{\min-\conv}(\vec u)$,
where $X_i, i\in N$ are replaced by their convex hulls
and the cost functions $\pi_i$
are replaced by their \emph{convex envelope}
or \emph{convex subfunctionals}.
With this convexification, it follows that the
duals of the original master problem~\ref{price-opt}
and that of the convexified game are equal.
With this insight, the characterization of enforceable
vectors $\vec u$ can (in some cases) be reduced to a more
tractable convex problem. 
The overall idea of convexifying a (nonconvex)
optimization problem is quite old and belongs to the broad field of \emph{global optimization}. Let us refer here to standard textbooks 
of the late seventies such as that of 
Horst and Tuy~\cite[\S 4.3.]{HorstT96} or Shapiro~\cite[\S 5]{shapiro1979}.
%.
%Also the fact that the dual of a convexified optimization problem
%is equal to the original one is well-known, see the above textbooks
For an overview on duality theory of general  non-convex programs, we refer to the work of Lemar{\'{e}}chal and Renaud~\cite{LemarechalR01}.

For the general approach to work, we need to make some mild assumptions.
\begin{assumption}\label{ass:convex}
We impose the following assumptions.
\begin{enumerate}
\item The strategy spaces $X_i\subset\R^m,i\in N$ are compact.
\item The functions $\vec x_i\mapsto \pi_i(\vec u,\vec x_i)$
are lower-semi-continuous (lsc)  on $X_i$ for  all $i\in N$.
\item\label{ass:concave-g} The functions $g_i(\vec x_i), i\in N$
are lsc and concave on $X_i$.\end{enumerate}
\end{assumption}
\begin{remark}
Condition~\eqref{ass:concave-g} includes the case that $g_i$ is linear, that is, $g_i(\vec x_i)=G_i\vec x_i$, where $G_i\in \R^{m\times m}$ is an $m\times m$  matrix.
\end{remark}
For $X_i\subset\R^m$ denote
$\conv(X_i):=\cap\{K\supset X_i | K\subset\R^m \text{ convex} \}$
the \emph{convex hull} of $X_i$ which by Assumption~\ref{ass:convex}
is closed and convex.
By the theorem of Carath\'{e}odory, every $\vec x_i\in \conv(X_i)\subset\R^m$
can be represented as a convex combination of
at most $m+1$ points in $X_i$.
We thus get
\[ \conv(X_i)=\left\{\sum_{k=1}^{m+1} \alpha_{ik} \vec y^k \middle\vert \vec y^k\in X_i, k=1,\dots,m+1, \bm\alpha_i\in\Lambda\right\},
 \]
where 
$\Lambda:= \{\bm \alpha\in \R_+^{m+1} \vert \vec 1^\intercal\bm\alpha=1\}$.
We now define the concept of a \emph{convex envelope}, see Horst and Tuy~\cite[\S 4.3.]{HorstT96}.
\begin{definition}
Let $K\subset\R^m$ be any compact set
and let $f:K\rightarrow \R$ be lsc. A \emph{convex envelope}
of $f$ on $\conv(K)$ is a function $\phi: \conv(K)\rightarrow \R$
satisfying:
\begin{enumerate}
\item $\phi$ is convex on $\conv(K)$.
\item $\phi(\vec x)\leq f(\vec x)$ for all $\vec x\in K$.
\item For all convex functions
$h: \conv(K)\rightarrow \R$ with $h(\vec x)\leq f(\vec x)$ for all $\vec x\in K$
we have $\phi(\vec x)\geq h(\vec x)$ for all $\vec x\in \conv(K)$.
\end{enumerate}
\end{definition}
From this definition it is evident that, if the convex envelope
exists, it is unique.
We will now explicitly describe the (unique) convex envelope
of the functions $\pi_i(\vec u,\vec x_i), i\in N$.
As shown by Grotzinger~\cite[Lemma 3.1.]{Grotzinger85}, under Assumption~\ref{ass:convex},
the convex envelopes of $\pi_i(\vec u,\vec x_i), i\in N$ exist and read as:\footnote{Rockafellar~\cite[pp.157]{Rockafellar70} showed
that without lsc and compactness of $X_i$, the convex envelope is given by the same formula
where $\min$ is replaced by $\inf$.} 
\begin{equation}\label{envelope}
\begin{aligned}
 \phi_i & :\R^m\times \conv(X_i) \rightarrow\R\\ (\vec u,\vec x_i) &\mapsto  
\min\left\{\sum_{k=1}^{m+1} \alpha_{ik} \pi_i(\vec u,\vec x_i^k)\middle\vert  \sum_{k=1}^{m+1}\alpha_{ik}\vec x_i^k=\vec x_i, \bm\alpha_i\in \Lambda, \vec x_i^k\in X_i, k=1,\dots, m+1\right\}.\end{aligned}
\end{equation}

\begin{definition}
For a game  $G^{\min}(\vec u)=(N, X, (\pi_i)_{i\in N})$,
the associated \emph{convexified game} is defined as
\[ \text{$G^{\min-\conv}(\vec u)=(N, X^{\conv}, (\phi_i)_{i\in N})$,
where  $X^{\conv}:=\times_{i\in N}\conv(X_i)$.}\]
\end{definition}
We obtain the following characterization result
connecting $G^{\min}(\vec u)$ with $G^{\min-\conv}(\vec u)$.

\begin{theorem}\label{thm:main-convexification}
Let  $\vec x^*\in X\subseteq X^{\conv}$ and $\bm\lambda\in \R_+^m$.
Then, under Assumption~\ref{ass:convex}, the following statements are equivalent.
\begin{enumerate}
\item\label{enum:main-convexification1}  $\vec u\in \R^m$ is enforceable
for $G^{\min}(\vec u)$ via $(\vec x^*,\bm\lambda)$.
\item\label{enum:main-convexification2} 
 $\vec u\in \R^m$ is enforceable for $G^{\min-\conv}(\vec u)$ via $(\vec x^*,\bm\lambda)$. 
\item\label{enum:main-convexification3} 
 $(\vec x^*,\bm\lambda)$ is a primal-dual optimal solution of~\ref{price-opt}  
for $G^{\min}(\vec u)$  and $\vec x^*$ satisfies
 $\ell(\vec x^*)=\vec u$.
 \item\label{enum:main-convexification4} 
 $(\vec x^*,\bm\lambda)$ is a primal-dual optimal solution of~\ref{price-opt}  
for  $G^{\min-\conv}(\vec u)$  and $\vec x^*$ satisfies
 $\ell(\vec x^*)=\vec u$.
\end{enumerate}
Moreover, all equivalences remain true by replacing the term ``enforceable''
with ``weakly enforceable by market prices'' and removing the condition
$\ell(\vec x^*)=\vec u$ in Statements~\ref{enum:main-convexification3}.
and~\ref{enum:main-convexification4}.
\end{theorem}
\begin{proof}
Since $G^{\min-\conv}(\vec u)$ fits into the framework
presented so far,  Theorem~\ref{thm:main}
implies already \ref{enum:main-convexification2}.~$\Leftrightarrow$~
\ref{enum:main-convexification4}.
and \ref{enum:main-convexification1}.~$\Leftrightarrow$~
\ref{enum:main-convexification3}.
Thus, we only need to show that the Problems~\ref{price-opt}
for $G^{\min-\conv}(\vec u)$ and $G^{\min}(\vec u)$, respectively,
have the same dual.
We get 
\begin{align} \notag & \mu^{\conv}(\bm\lambda)= \min_{\vec x\in X^{\conv}}\left\{\sum_{i\in N}\phi_i(\vec u,\vec x_i)+\bm\lambda^\intercal(\ell(\vec x)-\vec u)\right\}\\ \notag&=\sum_{i\in N}\min\left\{\sum_{k=1}^{m+1} \alpha_{ik} \pi_i(\vec u,\vec x_i^k)+\bm\lambda^\intercal g_i\left(\sum_{k=1}^{m+1}\alpha_{ik}\vec x_i^k\right)\middle\vert
 \bm\alpha_i\in \Lambda, \vec x_i^k\in X_i, k=1,\dots,m+1\right\}
 -\bm\lambda^\intercal\vec u
 \\
\label{eq:dual-convexified}&=\min_{\vec x_i\in X_i,i\in N}\left\{\sum_{i\in N}\pi_i(\vec u,\vec x_i)+\bm\lambda^\intercal(\ell(\vec x)-\vec u)\right\}\\
\notag&=\mu(\bm\lambda).
\end{align}
where~\eqref{eq:dual-convexified} follows
from the concavity of the inner objective functions w.r.t. $\bm\alpha_i, i\in N$. \footnote{For any compact set $\emptyset\neq S\subseteq \R^n$
and concave and lsc function $f:\R^n\rightarrow \R$, we have
$\min\{f(x)\vert x\in S\}=\min\{f(x)\vert x\in \conv(S)\}$.}
\end{proof}
Under a Slater-type constraint qualification and assuming that $g_i, i\in N$ are linear,
we get that~\ref{price-opt} for $G^{\min-\conv}(\vec u)$
always has zero duality gap leading to the following result.

\begin{theorem}\label{convexified:qualification}
Assume that  $g_i, i\in N$ are linear and there is
$\vec x^0 \in \relint{(X^{\conv})}\cap\{\vec x\vert\ell(\vec x)\leq \vec u\}$.
Then, under Assumption~\ref{ass:convex}, the following two statements hold.
\begin{enumerate}
\item\label{enum:convexified-slater1} Any minimal $\vec u$ for $X^{\conv}$  is  enforceable  for $G^{\min-\conv}(\vec u)$.
\item\label{enum:convexified-slater2}  $\vec u$ is enforceable  for $G^{\min}(\vec u)$ if and only if
Problem~\ref{price-opt} for $G^{\min-\conv}(\vec u)$
admits an optimal solution $\vec x^*\in X$ with $\ell(\vec x^*)=\vec u$.
\end{enumerate}
Moreover, the equivalence in~\ref{enum:convexified-slater2}. remains true by replacing the term ``enforceable''
by ``weakly enforceable with market prices'' and removing the condition
$\ell(\vec x^*)=\vec u$.
\end{theorem}
We now discuss two important special cases of Theorem~\ref{thm:main-convexification}.
\subsection{Finite Point Sets and Concave Extensions}

  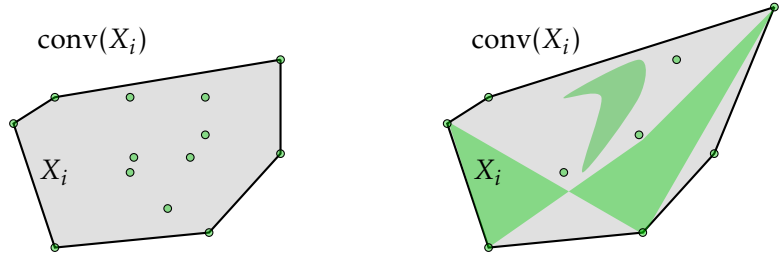
\begin{figure}[h!]
%  \begin{subfigure}[t]{0.4\textwidth}
 \begin{center}
           \begin{tikzpicture}
 \draw[thick, fill=gray!80!black,fill opacity=0.2]
(3.0,2.5) -- (3,1.25) -- (2.05,0.2) -- (0,0) -- (-0.55,1.65) -- (0,2) -- (3,2.5);

 \node at (0.5,2.75) {$\conv(X_i)$};
 \node at (0,1) {$X_i$};
   \draw[fill=green!80!black,fill opacity=0.4] (1,1) circle (.3ex);
      \draw[fill=green!80!black,fill opacity=0.4] (1,2) circle (.3ex);
        \draw[fill=green!80!black,fill opacity=0.4] (2,2) circle (.3ex);
     %     \draw[fill=green!80!black,fill opacity=0.4] (0,1) circle (.3ex);
       %     \draw[fill=green!80!black,fill opacity=0.4] (0,0) circle (.3ex);
         %     \draw[fill=green!80!black,fill opacity=0.4] (1,0) circle (.3ex);
              
               \draw[fill=green!80!black,fill opacity=0.4] (2,1.5) circle (.3ex);
            \draw[fill=green!80!black,fill opacity=0.4] (3,1.25) circle (.3ex);
   %           \draw[fill=green!80!black,fill opacity=0.4] (1,0) circle (.3ex);
      %          \draw[fill=green!80!black,fill opacity=0.4] (2,0) circle (.3ex);
      %        \draw[fill=green!80!black,fill opacity=0.4] (3,0) circle (.3ex);
               \draw[fill=green!80!black,fill opacity=0.4] (0,2) circle (.3ex);
             % \draw[fill=green!80!black,fill opacity=0.4] (0,3) circle (.3ex);
          
         %Convex Hull 
               \draw[fill=green!80!black,fill opacity=0.4] (3,2.5) circle (.3ex);
               
                  \draw[fill=green!80!black,fill opacity=0.4] coordinate(p1) circle (.3ex);
                  \draw[fill=green!80!black,fill opacity=0.4] (2.05,0.2) circle (.3ex);
                   \draw[fill=green!80!black,fill opacity=0.4] (1.05,1.2) circle (.3ex);
                    \draw[fill=green!80!black,fill opacity=0.4] (1.5,0.52) circle (.3ex);
                     \draw[fill=green!80!black,fill opacity=0.4] (1.8,1.2) circle (.3ex);
                    \draw[fill=green!80!black,fill opacity=0.4] (-0.55,1.65) circle (.3ex);
           \end{tikzpicture}\hspace{2cm}
   \begin{tikzpicture}
\path [fill=green!80!black,fill opacity=0.4] plot [smooth cycle] coordinates {(1,2) (1.5,2) (1.25,1) (2,2) (2,2.5)};
 \draw[thick, fill=gray!80!black,fill opacity=0.2]
(3.8,3.2) -- (3,1.25) -- (2.05,0.2) -- (0,0) -- (-0.55,1.65) -- (0,2) -- (3.8,3.2);
 \path[fill=green!80!black,fill opacity=0.4] (0,0) coordinate(p1) --  ++(35:2.5) coordinate(p2)
 -- ++(45:2.5) coordinate(p3) --
 ++(-120:3.5) coordinate(p4) --  ++(150:3) coordinate(p5);
 \node at (0.5,2.75) {$\conv(X_i)$};
 \node at (0,1) {$X_i$};
   \draw[fill=green!80!black,fill opacity=0.4] (1,1) circle (.3ex);
  %    \draw[fill=green!80!black,fill opacity=0.4] (1,2) circle (.3ex);
        \draw[fill=green!80!black,fill opacity=0.4] (2.5,2.5) circle (.3ex);
     %     \draw[fill=green!80!black,fill opacity=0.4] (0,1) circle (.3ex);
       %     \draw[fill=green!80!black,fill opacity=0.4] (0,0) circle (.3ex);
         %     \draw[fill=green!80!black,fill opacity=0.4] (1,0) circle (.3ex);
              
               \draw[fill=green!80!black,fill opacity=0.4] (2,1.5) circle (.3ex);
            \draw[fill=green!80!black,fill opacity=0.4] (3,1.25) circle (.3ex);
   %           \draw[fill=green!80!black,fill opacity=0.4] (1,0) circle (.3ex);
      %          \draw[fill=green!80!black,fill opacity=0.4] (2,0) circle (.3ex);
      %        \draw[fill=green!80!black,fill opacity=0.4] (3,0) circle (.3ex);
               \draw[fill=green!80!black,fill opacity=0.4] (0,2) circle (.3ex);
             % \draw[fill=green!80!black,fill opacity=0.4] (0,3) circle (.3ex);
          
         %Convex Hull 
               \draw[fill=green!80!black,fill opacity=0.4] (3.8,3.2) circle (.3ex);
               
                  \draw[fill=green!80!black,fill opacity=0.4] coordinate(p1) circle (.3ex);
                  \draw[fill=green!80!black,fill opacity=0.4] (2.05,0.2) circle (.3ex);
                    \draw[fill=green!80!black,fill opacity=0.4] (-0.55,1.65) circle (.3ex);
           \end{tikzpicture}
 
         \caption{Left is the scenario of $X_i$ consisting of a finite
point set. Right, $X_i$ may consist of connected components (in green) 
and isolated points but the convex hull is assumed to be finitely generated  and additionally $\pi_i(\vec u,\vec x_i)$ is assumed to be concave on $\conv(X_i)$. }\label{fig:convex-hull}
 %\end{subfigure}
 \end{center}
\end{figure}

We  now consider two special cases: in the first one, $X_i, i\in N$
consists of a finite collection of points (see Fig.~\ref{fig:convex-hull} left)
and in the second one, we assume that the convex hull of each  $X_i, i\in N$ is assumed to be finitely generated and  $\pi_i(\vec u,\vec x_i), i\in N$ is  concave on $\conv(X_i)$ (see Fig.~\ref{fig:convex-hull} right).
We start with the first model.
\begin{assumption}\label{ass:convex-hull}
For all $i\in N$, $X_i=\{\vec{{x}}^1_i,\dots, \vec{{x}}_i^{{k}_i}\}$
for some $k_i\in \N$. %\end{equation} 
\end{assumption}
With this assumption, the convex envelope
$\phi_i, i\in N$ has a simple form.
Let us define the following optimal value function of an
associated LP:
\begin{equation}\label{envelope:finite-point} \phi^{\LP}_i(\vec x_i):= 
\min\{ \bm \pi_i^\intercal \bm\alpha_i\vert \mathcal{X}_i \bm\alpha_i=\vec x_i, \bm\alpha_i\in\Lambda_i\},
\end{equation}
where 
$\bm \pi_{i}:=(\pi_i(\vec u,\vec{{x}}_i^k))_{k\in \{1,\dots,k_i\}}$, $\mathcal{X}_i := (\vec x_i^1,\cdots, \vec x_i^{k_i})$
is a $m\times k_i$ matrix with columns $\vec x_i^k, k=1,\dots, k_{i}$,
and $\Lambda_i:=\{\alpha_i\in \R^{k_i} \vert \vec 1^\intercal \bm\alpha_i=1, \bm\alpha_i\geq 0\}$.

\begin{lemma}\label{lem:finite}
Under Assumption~\ref{ass:convex} and Assumption~\ref{ass:convex-hull},
it holds that $\phi_i(\vec u,\vec x_i)=  \phi^{\LP}_i(\vec x_i)$
for all $\vec x_i\in \conv(X_i), i\in N$, where
$\phi_i, i\in N$ is the convex envelope as defined in~\eqref{envelope}.
\end{lemma}
\begin{proof}
The inequality $\phi_i^{\LP}(\vec x_i)\leq \phi_i(\vec u,\vec x_i) $ follows
directly as for any  $\vec x_i=\sum_{k=1}^{m+1} \alpha_{ik}  \tilde{x}_i^k$,
with $ \tilde{x}_i^k\in X_i, k= 1,\dots,m+1$,
 the corresponding $\bm\alpha_i$
is feasible for~\eqref{envelope:finite-point} with the same objective value.

For $\phi_i^{\LP}(\vec x_i)\geq \phi_i(\vec u,\vec x_i)$, 
we need to show that the LP has optimal solutions
with support less or equal than $m+1$.
The polytope $P_i:=\{\alpha_i\in \R^{k_i} \vert \mathcal{X}_i \bm\alpha_i=\vec x_i,  \vec 1^\intercal \bm\alpha_i=1, \bm\alpha_i\geq 0\}$
of the LP is non-empty and in standard form.
With the
theorem of linear programming we get that
an optimal solution of the LP is attained
at a vertex of $P_i$. Any vertex $\bm\alpha_i$ of $P_i$ has the property that
the columns of the defining matrix 
corresponding to indices $j$ with $\alpha_{ij}>0$ 
are linearly independent. Since this matrix has $m+1$ rows, its
rank is less than $m+1$ implying the wanted small support representation.
\end{proof}

We discuss now another class of \emph{concave problems}
for which we also get an LP
representation of the convex envelope.
\begin{assumption}\label{ass:concave}
The sets $X_i\subset\R^m, i\in N$ satisfy 
$\conv(X_i)=\conv\big(\{\vec{{x}}^1_i,\dots, \vec{{x}}_i^{{k}_i}\}\big)$
with  $\vec{{x}}^j_i\in X_i$ for $j=1,\dots, k_i, k_i\in \N,$
and
the functions $\pi_i(\vec u, \vec x_i), i\in N$ can be extended to
the domain $\conv(X_i)$ so that they are concave on $\conv(X_i)$.
\end{assumption}
With this assumption, the function $\phi_i^{\LP}(\vec x_i)$
defined in~\eqref{envelope:finite-point} is also equal to $\phi_i(\vec u,\vec x_i)$.

\begin{lemma}\label{lem:concave}
Under Assumption~\ref{ass:convex} and Assumption~\ref{ass:concave},
it holds that $\phi_i(\vec u,\vec x_i)=  \phi^{\LP}_i(\vec x_i)$
for all $\vec x_i\in \conv(X_i), i\in N$.
\end{lemma}
\begin{proof}
$\phi_i(\vec u,\vec x_i)\leq  \phi^{\LP}_i(\vec x_i)$
follows from the second part of the proof of the previous Lemma~\ref{lem:finite}.

For the other direction, 
let $\vec x_i\in  \conv(X_i)$
with  $\phi_i(\vec u,\vec x_i)=\sum_{k=1}^{m+1} \alpha_{ik}\pi_i(\vec u,\vec{y}_i^k)$ for
  $\vec x_i=\sum_{k=1}^{m+1} \alpha_{ik}\vec y_i^k$
with $\vec y_i^k\in X_i$ for  $k=1,\dots, m+1$ and $\alpha_i\in\Lambda$.
We  first bound the cost of every summand $\pi_i(\vec u,\vec{y}_i^k)$ individually. With $\vec{y}_i^k\in \conv(X_i)$ we have
$ \vec{y}_i^k= \sum_{j=1}^{k_i} \kappa_{ij}^k \vec{{x}}_i^j$
for  some $\bm \kappa_{i}^k\in \Lambda_i$.
With the concavity of $\pi_i(\vec u, \vec x_i)$, we get
\begin{align}\label{eq:individual}
\pi_i(\vec u,\vec{y}_i^k) = \pi_i\Big(\vec u,\sum_{j=1}^{k_i} \kappa_{ij}^k \vec{{x}}_i^j\Big) \geq \sum_{j=1}^{k_i} \kappa_{ij}^k \pi_i\big(\vec u,\vec{{x}}_i^j\big) \geq \phi_i^{LP}(\vec y_i^k),
\end{align}
where we use that $ \kappa_{i}^k$ is feasible for the LP associated with $\phi_i^{LP}(\vec y_i^k)$.
Let us write
$\phi_i^{LP}(\vec y_i^k)=\pi_i^\intercal \beta_i^k $ for some 
$ \beta_i^k\in\arg\min\{ \bm \pi_i^\intercal \bm\alpha_i\vert \mathcal{X}_i \bm\alpha_i=\vec y_i^k, \bm\alpha_i\in\Lambda_i\}.$
We then get
\begin{align*}
\phi_i(\vec u,\vec x_i)&=\sum_{k=1}^{m+1} \alpha_{ik}\pi_i(\vec u,\vec{y}_i^k)
\underset{\eqref{eq:individual}}{\geq} \sum_{k=1}^{m+1} \alpha_{ik}\phi_i^{LP}(\vec y_i^k)
= \sum_{k=1}^{m+1} \alpha_{ik}\pi_i^\intercal \beta_i^k
=\pi_i^\intercal  \sum_{k=1}^{m+1} \alpha_{ik} \beta_i^k
\geq \phi_i^{LP}(\vec x_i),
\end{align*}
where we used for the last inequality that
the vector $\sum_{k=1}^{m+1} \alpha_{ik} \beta_i^k$
is feasible for the LP corresponding to $\phi_i^{LP}(\vec x_i)$.
To see this, observe
\[ \mathcal{X}_i \sum_{k=1}^{m+1} \alpha_{ik} \beta_i^k=
  \sum_{k=1}^{m+1} \alpha_{ik} \mathcal{X}_i \beta_i^k=
 \sum_{k=1}^{m+1} \alpha_{ik}  \vec y_i^k=\vec x_i,
 \text{ and } \vec 1^\intercal \sum_{k=1}^{m+1} \alpha_{ik} \beta_{i}^k=\sum_{k=1}^{m+1} \alpha_{ik} \vec 1^\intercal \beta_{i}^k=\sum_{k=1}^{m+1} \alpha_{ik}=1.\]
\end{proof}

\subsection{The Master LP and its Dual}

Now we will 
model~\ref{price-opt} for $G^{\min-\conv}(\vec u)$
for any of the two previous game classes
via the following LP in the
variables $\bm \alpha_i, i\in N$. We assume from
now on that $g_i, i\in N$ are linear, that is, $g_i(\vec x_i)=G_i\vec x_i$, where $G_i\in \R^{m\times m}$ is an $m\times m$  matrix.
\begin{framed}
\begin{align}\tag{LP$^{\min}$($\vec u$)}\label{LP-lambda-conv}
\min  \sum_{i\in N} \bm \pi_{i}^\intercal \bm \alpha_i & \\
%\bm \pi_{i}^\intercal \bm \alpha_i\\
%\sum_{h\in [k_i]}\pi_i(\vec u, \vec{{x}}_i^l) \alpha_{ih}\\
\label{ineq:u-convex-finite} \ell(\bm\alpha)&\leq \vec u,\\ \notag \bm\alpha_i&\in\Lambda_i \text{ for all } i\in N,
  %\sum_{h\in [k_i]} \alpha_{ih}&= 1, \text{ for all }i\in N \\
%\bm \alpha_i & \geq 0\text{ for all } i\in N.
\end{align}
where
$\ell(\bm \alpha):=\sum_{i\in N}  \sum_{k\in \{1,\dots,k_i\}} \alpha_{ik} g_i(\vec{{x}}_{i}^k)$ and $\Lambda_i:= \{\bm \alpha_i\in \R_+^{k_i} \vert \vec 1^\intercal\bm\alpha_i=1\}, i\in N.$
\end{framed}

\begin{theorem}\label{thm:convex-hull-finite}
Let $G^{\min}(\vec u)$ be a game for which 
Assumptions~\ref{ass:convex} and \ref{ass:convex-hull} hold
and assume that $g_i, i\in N$ are linear.
Then, the following statements are equivalent
for the respective \ref{LP-lambda-conv}.
\begin{enumerate}
\item\label{enum:integer1} The capacity vector $\vec u\in \R^m$ is enforceable for $G^{\min}(\vec u)$.
\item\label{enum:integer2}  \ref{LP-lambda-conv} admits
an integral optimal solution $\bm\alpha^*$ for which~\eqref{ineq:u-convex-finite} is tight.
\end{enumerate}
Let $G^{\min}(\vec u)$ be a game for which 
Assumptions~\ref{ass:convex} and \ref{ass:concave} hold true
and assume $g_i, i\in N$ to be linear.
Then, the following statements are equivalent
for the respective \ref{LP-lambda-conv}.
\begin{enumerate}
\setcounter{enumi}{2}
\item\label{enum:integer3}  The capacity vector $\vec u\in \R^m$ is enforceable for $G^{\min}(\vec u)$.
\item\label{enum:integer4}  \ref{LP-lambda-conv} admits an
optimal solution $\bm\alpha^*$ with $\sum_{j=1}^{k_i} \alpha_{ij}^* \vec{{x}}_i^j\in X_i, i\in N$ with~\eqref{ineq:u-convex-finite} being tight.
\end{enumerate}
Moreover, the equivalence remains true by replacing the term ``enforceable''
with ``weakly enforceable by market prices'' and removing the condition
$\ell(\bm\alpha^*)=\vec u$ in Statements~\ref{enum:integer2}. and~\ref{enum:integer4}.
\end{theorem}
 \begin{proof}
By Lemmata~\ref{lem:finite} and~\ref{lem:concave} and the linearity of $g_i, i\in N$, the \ref{LP-lambda-conv}
is a correct formulation of Problem~\ref{price-opt} for 
the respective convexified games $G^{\min-\conv}(\vec u)$. 
Hence, the result follows from Theorem~\ref{thm:main-convexification}.
\end{proof}
\begin{remark}
The equivalence between~\eqref{enum:integer1} and~\eqref{enum:integer2}
remains true even for concave $g_i, i\in N$, since in this
case~\ref{LP-lambda-conv} is a relaxation of~\ref{price-opt}
w.r.t. the convexified game $G^{\min-\conv}(\vec u)$ with the same dual.
\end{remark}
\begin{remark}
The game $G^{\min-\conv}(\vec u)$ (assuming Assumption~\ref{ass:convex-hull}) can be interpreted as the mixed extension of $G^{\min}(\vec u)$.
By Theorem~\ref{thm:main}, this implies that for any finite strategic game
of type $G^{\min}(\vec u)$, any minimal vector $\vec u$ can be enforced 
in mixed strategies for the game $G^{\min-\conv}(\vec u)$
(the LP has zero duality gap).
Theorem~\ref{thm:convex-hull-finite}  implies that whenever
\ref{LP-lambda-conv} admits
integral optimal solutions, $\vec u$ is also enforceable in pure strategies.
\end{remark}

\ref{LP-lambda-conv} may in general involve (exponentially) many variables $\bm\alpha_i, i\in N$ depending on the number $k:=\sum_{i\in N}k_i$.
A common approach is to dualize~\ref{LP-lambda-conv}
to yield an LP
with less variables at the cost of obtaining (exponentially) many constraints. 
In the following we dualize the primal problem in the form
$  - \max  \{ - \sum_{i\in N} \bm \pi_{i}^\intercal \bm \alpha_i\vert 
\ell(\bm\alpha)\leq \vec u,\; \bm\alpha_i\in\Lambda_i,\;i\in N\}
$.
The following
steps are reminiscent to the standard dual LP
of the Walrasian configuration LP (see e.g., Blumrosen and Nisan~\cite[$\S$ 11.3.1]{Nisan:2007}).
\begin{framed}
\begin{align}\tag{DP$^{\min}$($\vec u$)}\label{convex:DP}
\min&  \sum_{i\in N} \mu_i+\sum_{j\in E}\lambda_j u_j  ,\\ \notag
\sum_{j\in E} g_{ij}(\vec x_{i}^k)  \lambda_j+\mu_i&\geq - \pi_{ik} \text{ for all }i\in N, k=1,\dots, k_i\\
 \mu_i&\in \R, i\in N,\; \lambda_j \geq 0, j\in E. \notag
\end{align}
\end{framed}
%Note that in order to obtain~\ref{bilateral:DP}, we have relaxed
%$\Lambda_i$ to the set $\{\bm\alpha_i\in \R^{k_i}\vert \sum_{j\in[k_i]}%\alpha_{ij}\leq 1\}$ which is feasible since any primal optimal solution
%will satisfy $\sum_{j\in[k_i]}\alpha_{ij}= 1$.
Note that $\mu_i,i\in N$ is not sign-constrained
as it is the dual variable to $\sum_{k}\alpha_{ik}= 1, i\in N$.
Moreover, recall that
$g_{ij}(\vec x_{i}^k) \in \R$ are just parameters in \ref{convex:DP}.
The dual has $n+m$ many variables but exponentially
many constraints, but,  if we have a polynomial
time separation oracle, we can use the ellipsoid method
to obtain a polynomial time algorithm (cf. Groetschel et al.~\cite{GroetschelLovaszSchrijver1993}).
A standard way to obtain such an oracle is to assume an efficient \emph{demand oracle}.
\begin{definition}\label{def:oracle}
A demand oracle for player $i\in N$ gets as input prices $\bm \lambda\in \R_+^m$
and outputs a cost minimizing vector $\vec x_i\in X_i$, that is,
%\begin{equation}\label{oracle-concave} 
\[\vec x_i(\bm\lambda)\in \arg\min\left\{\pi_i(\vec u, \vec x_i)+\bm\lambda^\intercal  g_i(\vec x_i) \vert \vec x_i\in X_i\right\}.
\]
%\end{equation}
\end{definition}

We obtain the following result
for polynomial time computable demand oracles.
Let us remark here that we assume that 
there is a succinct representation of the game $G^{\min}(\vec u)$
and hence of the \ref{LP-lambda-conv}. 
%For instance,  
%assume we a polynomial computable oracle access to the values $g_{ij}(\vec x_{i}^k), \pi_{ik},j\in E,  k\in \{1,\dots,k_i\},i\in N$. 
\begin{theorem}\label{concave-main}
Let $G^{\min}(\vec u)$ be a game for which the assumptions of Theorem~\ref{thm:convex-hull-finite} are satisfied.
If for all $\bm\lambda\in \R^m_+$ and $i\in N$,
the demand oracle $\vec x_i(\bm\lambda)$ can be computed in polynomial time,
then, \ref{LP-lambda-conv} can be solved in polynomial time.
\end{theorem}
\begin{proof}
In order to use the ellipsoid method, we need to check whether we get a polynomial time separation oracle
for the constraints:
\[ \sum_{j\in E} g_{ij}(\vec x_{i}^k) \lambda_j+\mu_i\geq -\pi_{ik}, i\in N, k=1,\dots, k_i. \]
With the demand oracle we can determine
the value $\pi_i^*(\bm\lambda):=\pi_i(\vec u, \vec x_i(\bm\lambda))+ \bm\lambda^\intercal g_i(\vec x_i(\bm\lambda))$.
Now,  if $\pi_i^*(\bm\lambda)\geq - \mu_i$ for all $i\in N$, the current point $(\bm\mu,\bm\lambda)$
is feasible. Otherwise, suppose
$\pi_i^*(\bm\lambda)< - \mu_i$.
If   $G^{\min}(\vec u)$ is a game with finite point sets $X_i$, 
we get of course $\vec x_i(\bm\lambda)=\vec x_i^k$
for some $k\in\{1,\dots, k_i\}$ and, hence,
\[ \pi_i^*(\bm\lambda)=\pi_i(\vec u, \vec x_{i}^k)+\sum_{j\in E} g_{ij}(\vec x_{i}^k) \lambda_j=\pi_{ik}+
\sum_{j\in E} g_{ij}(\vec x_{i}^k) \lambda_j<-\mu_i,\]
represents a violated inequality.

If $G^{\min}(\vec u)$ is a concave game, then, by the linearity of $g_{i}$, the function
$\vec x_i\mapsto \sum_{j\in E} g_{ij}(\vec x_i) \lambda_j+\pi_i(\vec u, \vec x_i)$
is concave over $\conv(X_i)$ as well and attains  its minimum over $\conv(X_i)$ at an extreme point of $\conv(X_i)$, hence, in $X_i$.
 \end{proof}
 
\begin{remark}
All results in this section carry directly over
to the case of maximization problems.
In this case, a convex envelope becomes
a concave envelope, the concave functions $g_i, i\in N$ become convex and the assumption
of having a concave extension changes to
a convex extension. 
\end{remark}
\subsection{Consequences and Impossibility Results}
The characterization result in Theorem~\ref{thm:convex-hull-finite}
together with the assumption of a polynomial time
demand oracle can be used to establish non-existence
results based on complexity-theoretic assumptions like $P\neq NP$.
If the master problem~\ref{price-opt} (which is also
called the welfare maximization problem in some applications) 
is NP-hard but there is a polynomial demand oracle, then, assuming $P\neq NP$, the 
guaranteed (weak) enforceability (with market prices) of $\vec u$ is ruled out since
otherwise, we can just compute an integral optimal solution of~\ref{LP-lambda-conv}
in polynomial time (by solving the dual~\ref{convex:DP})
which corresponds to an optimal solution of the master problem.
This approach has been pioneered by Talgam-Cohen and Roughgarden~\cite{Roughgarden:2015}
for the case of pricing equilibria for Walrasian market settings (and generalizations thereof).
\section{Integral Problems and Compact Linear Relaxations}\label{sec:integral}
 We assume in the following
 that for every $i \in N$, the set $X_i$ is of the form
$ X_i=\left\{\vec x_i\in \Z^m\middle\vert  A_i\vec x_i\geq b_i\right\},$
where $A_i$ is a rational $k_i\times m$ matrix
and $\vec b_i\in \IQ^{k_i}$ is a rational vector.
Thus,  the combined set is given by
$X:=\left\{\vec x\in \Z^{n\cdot m} \middle \vert A_i\vec x_i\geq \vec b_i, i\in N\right\}.$
 We further assume \emph{linear resource consumption},
that is,
$ g_i(\vec x_i)=\vec x_i \text{ for all }i\in N.$
This assumption implies $\ell_j(\vec x)=\sum_{i\in N} x_{ij}$ for all $j\in E$.
 The private cost function
of a player is assumed to be \emph{quasi-separable} over the resources 
and depends only on the aggregated
load vector and the own load on the resource:
$ \pi_i(\vec u, \vec x_i)= \sum_{j\in E} \pi_{ij}(\vec u)\cdot  x_{ij},$
where
$ \pi_{ij}: \Z^m \rightarrow \R $ denotes the player-specific per-unit
cost on resource $j$ mapping a vector $\vec u$ to the reals.
 Then, problem~\ref{price-opt}
can be reformulated as an \emph{integer linear optimization problem}.
If the polyhedron  $\relaxation(X):=\left\{\vec x\in \R^{n\cdot m} \middle \vert A_i\vec x_i\geq b_i, i\in N\right\}$
is integral\footnote{A polyhedron $P\subset \R^r$ is \emph{integral},
if all its vertices are integral.}, then, it follows directly that~\ref{price-opt} 
has zero duality gap.
\subsection{Aggregation Polytopes and Total Dual Integrality}\label{sec:aggregation}
A powerful tool to recognize integrality of polyhedra
 is the notion of
 \emph{total-dual-integrality} (TDI) of linear systems (see Edmonds and Giles~\cite{Edmonds1977}).
 A rational system of the form
$A\vec z\geq \vec b$ with $A\in \IQ^{r\times m}$ and $\vec b\in \IQ^r$  is TDI,
if for every integral $\vec c\in \Z^{m}$, the dual
of 
$ \min\{ \vec c^\intercal\vec z \vert A\vec z\geq \vec b\} $
given by
$ \max\{ \vec z^\intercal\vec b \vert A^\intercal \vec z=  \vec c, \vec z\geq 0\} $
has an integral optimal solution (if the problem admits a finite optimal solution). It is known that for TDI systems, the corresponding polyhedron
is integral. A system $A\vec z\geq \vec b$ is \emph{box-TDI},
if the system $A\vec z\geq \vec b, w\leq\vec z\leq  \vec u$ is TDI
for all rational $w,\vec u$. A polytope is called box-TDI,
if it can be described by a box-TDI system.

 Now we assume that for all $i\in N$,  the matrices $A_i$
 are equal to some matrix $A\in \IQ^{r\times m}$.
We further assume that the cost functions are linear
and \emph{homogenous},
that is, they have the form
$\pi_i(\vec u, \vec x_i)= \sum_{j\in E} \pi_{j}(\vec u)\cdot  x_{ij},$
where
$ \pi_{j}: \Z^m \rightarrow \R $ denotes the resource-specific per-unit
cost on resource $j$ mapping a load vector $\vec u$ to the reals.
Instead of taking the Cartesian product
 of the LP-relaxations $\relaxation(X_i)$, we define an \emph{aggregation polytope}:
 \begin{equation}\label{eq:aggregation-polytope}
 P_N=\{\vec z\in \R^m\vert A\vec z\geq \sum_{i\in N}\vec b_i, \vec z\geq 0\}.
 \end{equation}
 This aggregated polytope seems only useful, 
 if it is box-TDI and any solution $\vec z$ can be
 decomposed into feasible strategies.
 This latter property is called the \emph{integer decomposition property} 
 (IDP). Formally,  $P_N$ has the IDP, if any integral optimal solution $\vec z\in P_N$
 can be decomposed into feasible integral vectors, that is, $\vec z= \sum_{i\in N}\vec z_i$ with $\vec z_i\in X_i$
 for all $i\in N$.
 We remark  that Kleer and Sch\"afer~\cite{KleerS17}
showed - in a different context - that polytopal congestion
games (see Section~\ref{sec:congestion-games} for a definition) with box-integral and IDP aggregation polytopes have nice
properties in terms of equilibrium computation and equilibrium welfare properties.
  Now we have everything together to state the following result.
 \begin{theorem}\label{thm:box-integral}
Assume that $P_N$ is box-TDI and satisfies IDP. Then for homogeneous linear cost functions $\pi_j, j\in E$, every $u\in \Z^m$
for which  $P_N\cap\{\vec y \vert \vec y\leq \vec u\}\neq\emptyset$
is weakly enforceable by market prices. Moreover under the same
condition, every minimal $\vec u$ w.r.t. $P_N$ is enforceable.
  \end{theorem} 
  \iffalse
  \begin{proof}
 Using the homogeneity of cost functions 
 and the IDP assumption of $P_N$,
 the LP-relaxation of~\ref{ip-price-opt-minimal} becomes
  \[ \min\left\{\sum_{j\in E}\pi_j(\vec u) y_j \;\middle\vert\;  \vec y\in P_N\cap\{\vec y \vert \vec y\leq \vec u\} \right\}.\]
 As  $P_N$ is box-TDI, there is an integral optimal solution
 which by IDP can be decomposed into the original sets $X_i,i\in N$.
 \end{proof}
\fi
\subsection{Integral Polymatroid Games}
We consider now a class of  games
based on \emph{polymatroids}
which rely on submodular functions defining structured capacity constraints
on subsets of resources.
An integral set function $f : 2^E \rightarrow \Z$ is \emph{submodular} if $f(U) + f(V) \geq f(U \cup V) + f(U \cap V)$ for all $U,V \in 2^E$; $f$ is \emph{monotone} if
$f(U)\le f(V)$ for all $U\subseteq V \subseteq E$; and $f$ is \emph{normalized} if $f(\emptyset)=0$.
We call an integral, submodular, monotone, and normalized function $f : 2^E \rightarrow \Z$ an \emph{integral polymatroid rank function}.
\iffalse
The associated \emph{integral polymatroid polyhedron} is defined as
$\P_f := \Bigl\{\vec y \in \Z^m \mid y(U) \leq f(U)\text{ for all } U\subseteq E \Bigr\}$,
where for a vector $\vec y = (y_j)_{j \in E}$ and $U \subseteq E$, we write $y(U)$ shorthand for $\sum_{j \in U} y_j$. For $\P_f$, the corresponding \emph{integral polymatroid base polyhedron} is defined as
$ \B_f := \Bigl\{\vec y\in \Z^m \mid y(U) \leq f(U)\text{ for all }U\subseteq E, \;y(E) = f(E)\Bigr\}.
$
We denote by 
$ \EB_f := \Bigl\{\vec y\in \R^m \mid y(U) \leq f(U)\text{ for all }U\subseteq E, \;y(E) = f(E)\Bigr\}.
$
the \emph{relaxed polymatroid base polyhedron}.

\fi

Suppose  there
is a finite set $N=\{1,\dots,n\}$ of players so that each player~$i$ is associated with an integral polymatroid rank function $f_i : 2^E \to \Z$ that defines an integral polymatroid $\P_{f_i}$ with base polyhedron $\B_{f_i}$. A strategy of
player~$i \in N$ is to choose a vector $\vec x_i  = (x_{ij})_{j \in E} \in \B_{f_i}$, i.e., player~$i$ chooses an integral resource consumption $x_{ij} \in \Z$ for each resource $e$
such that $f_i(E)$ units are distributed over the resources and for each $U \subseteq E$ not more than $f_i(U)$ units are distributed over the resources contained in $U$.
Formally, the set $X_i$ of feasible strategies of player~$i$ is
defined as
\begin{align*}
X_i = \B_{f_i}
= \Bigl\{\vec x_i \in \Z^m \mid x_{i}(U) \le f_i(U) \text{ for all } U \subseteq 
E,\; x_{i}(E) =f_i(E) \Bigr\},
\end{align*}
where, for a set $U \subseteq E$, we write $x_i(U)=\sum_{j \in U} x_{ij}$.
We show that the LP-relaxation
admits integral optimal solutions - by reformulating ~\ref{price-opt}
as a polymatroid intersection problem whose underlying intersection polytope is known
to admit integral optimal solutions.\footnote{The intersection of two
polymatroid base polytopes, however,  need not be a polymatroid.}
The proof can be found in the appendix.
\begin{theorem}\label{thm:polymatroid-main}
For polymatroid games, every $\vec u\in \Z^m$ 
for which $P^{\min}(\vec u)$ admits a finite optimal solution
 is weakly enforceable with market prices.
If $\vec u\in \Z^m$ is minimal for $X$, then $\vec u$ is  enforceable.
\end{theorem}
\iffalse
Before we prove the theorem, we state
some observations. It is known that the fractional relaxation of any individual integral 
polymatroid base polyhedron  $\EB_{f_i}$ is box-integral -- however this
does not imply that $\relaxation(X) \cap \{\vec x\in \R^{nm} \vert \ell(\vec x)\leq \vec u\} $ is also integral. Take for instance the path packing problem in a capacitated graph, where one needs to find
$k$ disjoint paths for source-terminal pairs $s_i$,$t_i, i=1,\dots,k$
in a digraph.
 For any individual $s_i$,$t_i$ pair,  the
flow polytope is integral and a feasible path can be computed by shortest path computations.
However, computing a path packing for multiple  $s_i$,$t_i$ pairs is strongly NP-hard
and does not admit a polynomially sized integral LP formulation (unless $P=NP$).

Nevertheless, polymatroids carry enough structure so that
the combined polytope described in $LP^{\min-\polymatroid}(\vec u)$ remains
integral. 
\fi

%The above proof is  that Bornd\"orfer~\cite{Borndoerfer2004} proved a %similar result for the special case of matroids and unit integral capacity %vectors $\vec u$.

For the maximization variant, we refer to Appendix~\ref{sec:polym-max}.

\section{Monotone Aggregative Games}\label{sec:monotone}

In the previous sections, we assumed that the private
cost function $\pi_i(\vec u, \vec x_i)$ of every player $i\in N$
is parameterized in $\vec u$ and, otherwise,
only depends on $\vec x_i\in X_i$.
This separability condition allowed to decompose
the Lagrangian of Problem~\ref{price-opt}
 leading to the subsequent characterizations.
Several games of interest, however, do not fulfil
this assumption. A prime example are
atomic congestion games, where the private
cost  depends on the aggregated load vector $\ell(\vec x)$ of
all players and changes, if player $i$ changes her
strategy to some $\vec y_i\neq\vec x_i\in X_i$.
We introduce a class of \emph{monotone aggregative
games (mag)}, where the cost function of a player $i\in N$ is allowed to
depend on both, $\vec x_i\in X_i$ and the aggregate $\ell(\vec x), \vec x\in X$. Throughout this section, we assume that $X_i\subset\R^m_+$
and $g_i:X_i\rightarrow\R^m_+$
for all $i\in N$. For $\vec x_i\in X_i$, we use the notation $E(\vec x_i ):=\{j\in E\vert x_{ij}>0\}$. 

\begin{assumption}\label{ass:mag}
Let $X_i\subseteq\R^m_+$
for all $i\in N$. We assume that cost/utility functions
depend on the vector of loads and on the
own strategy vector only -- this structure is
known in the literature as \emph{aggregative games}, see Harks and Klimm~\cite{HarksK15}, Jensen~\cite{Jensen10} and Paccagnan et al~\cite{PaccagnanGPKL19}.
\begin{enumerate}
\item\label{agg:cost} For minimization games $G^{\min-\moag}$, the total cost of a player $i\in N$ under strategy distribution $\vec x\in X$ is defined by a function $\cost_{i}^{\moag}:X\rightarrow \R,$
which satisfies 
\[ \cost_i^{\moag}(\vec x):= \pi_i^{\moag}(\ell(\vec x),\vec x_i) \text{ for all }\vec x\in X\]
for some function $\pi_i^{\moag}:\R_+^m\times X_i\rightarrow \R$.
For maximization games $G^{\max-\moag}$, we denote the utility function 
for $i\in N$ by 
$\utility_{i}^{\moag}:X\rightarrow \R$ and we assume that it satisfies 
\[ \utility_i^{\moag}(\vec x):= v_i^{\moag}(\ell(\vec x),\vec x_i) \text{ for all }\vec x\in X\]
for some function $v_i^{\moag}:\R^m_+\times X_i\rightarrow \R$.
\item\label{agg:mon} We further assume that the indirect cost/utility functions
exhibit \emph{negative externalities} in the load vector, that is, 
the following \emph{monotonicity condition} holds: \begin{align}\label{eq:monotonicty}
\pi_i^{\moag}(\vec u,\vec x_i)\leq \pi_i^{\moag}(\vec w,\vec x_i) \text{ for all }\vec u\leq \vec w, \vec u,\vec w\in \R_+^m\\ \notag
v_i^{\moag}(\vec u,\vec x_i)\geq v_i^{\moag}(\vec w,\vec x_i) \text{ for all }\vec u\leq \vec w, \vec u,\vec w\in \R_+^m.
\end{align}
\item The indirect cost/utility functions satisfy the condition of
 \emph{independence of irrelevant alternatives} defined as follows: 
\begin{align}\label{eq:independent}  \pi_i^{\moag}(\vec u,\vec x_i)&= \pi_i^{\moag}(\vec w,\vec x_i)
\text{ for all } \vec u,\vec w\in \R_+^m\text{ with } u_{j}=w_{j} \text{ for all }j\in E(\vec x_i), \vec x_i\in X_i,i\in N  \\ \notag
v_i^{\moag}(\vec u,\vec x_i)&= v_i^{\moag}(\vec w,\vec x_i)
\text{ for all } \vec u,\vec w\in \R_+^m\text{ with } u_{j}=w_{j} \text{ for all }j\in E(\vec x_i), \vec x_i\in X_i,i\in N.
\end{align}
\item\label{ass:mag-strategy} The strategy spaces $X_i\subseteq\R^m_+, i\in N$ exhibit an \emph{overlapping structure}, that is, for all $i\in N$:
\begin{equation}\label{eq:laminar}
g_{ij}(\vec x_i)=g_{ij}(\vec y_i) \text { for all } j\in E(\vec x_i)\cap E(\vec y_i),
\vec x_i,\vec y_i \in X_i.
%E  \text { with } x_{ij}>0, y_{ij}>0.
\end{equation}
\end{enumerate}
We call a game satisfying~\ref{agg:cost}.-\ref{ass:mag-strategy}.
a \emph{monotone aggregative game}.
\end{assumption}
Condition~\ref{ass:mag-strategy}. implies that for any 
resource $j\in E$ with $j\in E(\vec x_i)\cap E(\vec y_i)$
for at least two $\vec x_i\neq\vec y_i\in X_i$, the  resource consumption level of 
player $i$ on resource $j$ is fixed to some $\kappa_{ij}\geq 0$
for every strategy $\vec x_i\in X_i$ with $j\in E(\vec x_i)$.
On every other resource $j\in E$, however, the resource
usage level may be arbitrary.
We need to redefine the concept of enforceability
for this more general class of games.

\begin{definition}[Enforceability]\label{def:enforceable-ant}
A vector $\vec u\in \R^m$ is enforceable by prices $\bm\lambda\in\R^m_+$, if
there is $\vec x^*\in X$ satisfying~\ref{cond1-ant}. and~\ref{cond2-ant}.
for minimization games $G^{\min-\moag}$ or~\ref{cond1-ant}. and~\ref{cond3-ant}.
for maximization games $G^{\max-\moag}$:
\begin{enumerate}
\item\label{cond1-ant}  $\ell_j(\vec x^*)= u_j$ for all $j\in E$.
\item Minimization: \label{cond2-ant} $ \vec x^*_i\in\arg\min_{\vec x_i\in X_i}\left\{ \pi_{i}(\ell(\vec x_i,\vec x_{-i}^*),\vec x_i)+\bm\lambda^\intercal g_i(\vec x_i)\right\} \text{ for all }i\in N.$
\item Maximization: \label{cond3-ant} $ \vec x^*_i\in\arg\max_{\vec x_i\in X_i}\left\{ v_i(\ell(\vec x_i,\vec x_{-i}^*),\vec x_i)-\bm\lambda^\intercal g_i(\vec x_i)\right\} \text{ for all }i\in N.$
\end{enumerate}
 \end{definition}

We obtain the following result for monotone aggregative games. 

\begin{theorem}\label{thm:mag}
Let $G^{\min-\moag}=(N,X,(\pi_i^{\moag})_{i\in N})$ be a monotone aggregative minimization game and let  $G^{\min}(\vec u)=(N,X,(\pi_i)_{i\in N})$ be an associated game with $\pi_i(\vec u,\vec x_i):=\pi_i^{\moag}(\vec u,\vec x_i), \vec x_i\in X_i, i\in N$. Denote the same assumption for a maximization game by exchanging $\min$
with $\max$ and $\pi$ with $v$.  Then, the following holds true.
\begin{enumerate}
\item\label{mag:min} If  $(\vec x^*, \bm\lambda)\in X\times \R_+^m$ enforces $\vec u$ for $G^{\min}(\vec u)$, then $(\vec x^*, \bm\lambda)$
also enforces $\vec u$ for $G^{\min-\moag}$.
\item\label{mag:max} If  $(\vec x^*, \bm\lambda)\in X\times \R_+^m$ enforces $\vec u$ for $G^{\max}(\vec u)$, then $(\vec x^*, \bm\lambda)$
also enforces $\vec u$ for $G^{\max-\moag}$.
\end{enumerate}
\end{theorem}
\begin{proof}
We only prove~\ref{mag:min}. as ~\ref{mag:max}. follows by the same arguments.
For~\ref{mag:min}., we only need to verify the Nash equilibrium conditions
of a given tuple $(\vec x^*, \bm\lambda)\in X\times \R_+^m$
for the game $G^{\min-\moag}$.
We obtain for any $\vec y_i\in X_i, i\in N$:
\begin{align}\notag
\pi_i^{\moag}(\ell(\vec x^*),\vec x_i^*)+\bm\lambda^\intercal g_i(\vec x_i^*) &= 
\pi_i(\vec u,\vec x_i^*)+\bm\lambda^\intercal g_i(\vec x_i^*)\\\label{ineq:x-opt-ant}
&\leq \pi_i(\vec u,\vec y_i)+\bm\lambda^\intercal g_i(\vec y_i)\\ \label{ineq:irr1}
& = \pi_i^{\moag}(\vec w^1,\vec y_i)+\bm\lambda^\intercal g_i(\vec y_i) \text{ for }
w^1_j:=\begin{cases}u_j,& \text{ for }j\in E(\vec y_i)\\
 0,& \text{ otherwise}\end{cases} \\ \label{ineq:monotone-ant}
 &\leq \pi_i^{\moag}(\vec w^2,\vec y_i)+\bm\lambda^\intercal g_i(\vec y_i) \text{ for }
w^2_j:=\begin{cases}w^1_j+g_{ij}(\vec y_i),& \text{ for }j\in E(\vec y_i)\setminus E(\vec x_i^*) \\
 w^1_j,& \text{ otherwise}\end{cases} \\ \label{ineq:irr2}
 &= \pi_i^{\moag}(\ell((\vec x_{-i}^*,\vec y_i)),\vec y_i)+\bm\lambda^\intercal g_i(\vec y_i),
\end{align}
where~\eqref{ineq:x-opt-ant} follows, because $u$ is enforced by $(\vec x^*, \bm\lambda)$ for $G^{\min}(u)$.
Equality~\eqref{ineq:irr1}
 follows from the independence of
irrelevant choice condition~\eqref{eq:independent}.
Inequality~\eqref{ineq:monotone-ant} follows
from the monotonicity condition~\eqref{eq:monotonicty}.
Equality~\eqref{ineq:irr2} follows from the overlapping condition~\eqref{eq:laminar}.
\end{proof}

\section{Applications in Congestion Games}\label{sec:congestion-games}
We now demonstrate the applicability
of our framework by deriving
new existence results of  tolls enforcing certain load vectors in congestion games.
Moreover, we show 
how several known results in the literature
follow directly.
\subsection{Nonatomic Congestion Games}
We first present results for the case
that the strategy spaces of players are convex subsets of $\R^m_+$.
%\paragraph{Nonatomic Routing.}
We are given a directed graph $G=(V,E)$
and a set of populations $N:= \{1, \dots,
n\}$, where each population $i \in N$ has a demand $d_i> 0$ that has
to be routed from a source $s_i \in V$ to a destination $t_i \in V$.
In the \emph{nonatomic} model, the demand interval $[0,d_i]$ represents a continuum of infinitesimally small agents each acting independently
choosing a cost minimal $s_i$,$t_i$ path. 
There are continuous cost  functions $c_{ij}: \R_+^m \rightarrow \R_+, i\in N, j\in E$  which may depend on the population identity
and also on the aggregate load vector -- thus allowing
for modeling \emph{non-separable} latency functions.
A \emph{flow} for population~$i\in N$ is a nonnegative vector
$\vec x_i \in \R^{|E|}_+$ that  lives in the flow polytope: 
\begin{align*}
X_i=\left\{\vec x_i\in \R_+^m\middle \vert \sum_{j\in \delta^+(v)} x_{ij} - \sum_{j\in \delta^-(v)} x_{ij} = \gamma_i(v), \text{ for all } v\in V\right\},
\end{align*}
where $\delta^+(v)$ and $\delta^-(v)$ are the arcs leaving and
entering~$v$, and $\gamma_i(v) = d_i$, if $v =
s_i$, $\gamma_i(v) = -d_i$, if $v = t_i$, and $\gamma_i(v)=0$, otherwise.
We assume that every $t_i$ is reachable in $G$ from $s_i$
for all $i\in N$, thus, $X_i\neq \emptyset$ for all $i\in N$.
Given a combined flow $\vec x\in X$, the cost of a path $P\in\P_i$,
where $\P_i$ denotes the set of simple $s_i,t_i$ paths in $G$, 
 is defined as
\[c_{i,P}(\ell(\vec x)):=\sum_{j\in P}c_{ij}(\ell(\vec x)).\]
A \emph{Wardrop equilibrium} $\vec x$ with path-decomposition $(x_{i,P})_{i\in N, P\in \P_i}$ is defined as follows:
\[ c_{i,P}(\ell(\vec x))\leq c_{i,Q}(\ell(\vec x)) \text{ for all }P,Q\in \P_i \text{ with }x_{i,P}>0.\]
The interpretation here is that any 
agent is traveling along a shortest path given the overall load vector $\ell(\vec x)$.
One can reformulate the Wardrop equilibrium conditions
using load vectors $\vec u$ stating
that - given the load vector of a Wardrop equilibrium -
every agent is traveling along a shortest path.
\begin{lemma}[Dafermos~\cite{dafermos1980traffic,dafermos71}]
A strategy distribution $\vec x^*\in X$ with overall load vector $\vec u:=\ell(\vec x^*)$ is a Wardrop equilibrium if and only if

\[ \vec x_i^*\in\arg\min\left\{\sum_{j\in E}c_{ij}(\vec u) x_{ij} \middle\vert\; \vec x_i\in X_i\right\} \text{ for all }i\in N.\]

%\[  c_{i,P}(\vec u)\leq c_{i,Q}(\vec u) \text{ for all }P,Q\in \P_i \text{ with }x_{i,P}>0.\]
\end{lemma}
With this characterization, the model fits in our framework and we can apply our general existence result.

\begin{corollary}[Yang and Huang~\cite{Yang04}, Fleischer et al.~\cite{Fleischer04}, Karakostas and Kolliopoulos~\cite{Karakostas04}]\label{cor:nonatomic}
Every minimal capacity vector $\vec u$ is enforceable.
\end{corollary}
\begin{proof}
Define 
$\pi_i(\vec u,\vec x_i):=\sum_{j\in E} c_{ij}(\vec u) x_{ij}$ 
and $g_i(\vec x_i):=\vec x_i$
for every $i\in N$.
By the linearity of the objective in~\ref{price-opt}
and the fact that Slater's constraint qualification condition is satisfied,
the result follows by Corollary~\ref{cor:convex}.
%\footnote{
%$X$ consists of the Cartesian product of the $X_i$'s that are convex non-empty flow polytopes. This set is intersected with a polyhedron defined by affine inequalities.}
\end{proof}
Note that the above result
is more general than that of~\cite{Fleischer04,Karakostas04,Yang04} as we allow for \emph{arbitrary} player-specific cost functions $c_{ij}, i\in N, j\in E$.
%However, for enforcing $\vec u$ via a \emph{Nash equilibrium},
%one needs to assume that latency functions are nondecreasing.
Previous works assumed less general
\emph{heterogeneous cost functions} of the form
\[ c_{i,P}(\vec x)=\sum_{j\in P } \alpha_i c_j(\ell_j(\vec x))+ \lambda_j,\]
where $\alpha_i>0$ represents a tradeoff parameter 
weighting the impact of money versus travel time.
Fleischer et al.~\cite[Sec. 6]{Fleischer04}
also mention that their existence result
holds for the more general case of non-separable
latency functions. 
\subsection{Atomic Congestion Games}\label{sec:atomic-congestion-games}
Now we turn to atomic congestion
games as introduced in Rosenthal~\cite{Rosenthal73}.
This setting arises by assuming that $d_i=1$
for all $i\in N$ and requiring
that all strategy vectors $\vec x_i$ need
to be integral, that is, $\vec x_i\in \{0,1\}^m$.
As is standard in the congestion games
literature, instead of considering \emph{network games},
where the strategies are paths
in graphs, we can 
associated a set $\Sc_i\subset 2^m$
of allowable subsets of $E$ with every $i\in N$
and the incidence vectors of a set $S\in \Sc_i$
represent a feasible $\vec x_i$ and vice versa. 
The cost functions on resources are given by \emph{player-specific  functions} $c_{ij}(\ell_j(\vec x)), j\in E, i\in N$.
If the cost functions only depend on the resource identity, that is, 
$c_{j}(\ell_j(\vec x)), j\in E$, we speak of an atomic
congestion game with \emph{homogeneous cost functions}.
One can easily verify that for monotonically nondecreasing (player-specific) cost functions,
the atomic congestion game is in fact a monotone aggregative
game.
\begin{lemma}\label{congestion-mag}
Every atomic congestion game with monotonically nondecreasing (player-specific) cost functions
is a monotone aggregative game.
\end{lemma}
We leave the proof to the reader.
With this result at hand,
by Theorem~\ref{thm:mag} it suffices 
to analyze cases for which
$\vec u$ is enforceable for $G^{\min}(\vec u)$,
which in turn is equivalent to the property
that Problem~\ref{price-opt} has zero duality gap.
Note that Problem~\ref{price-opt} has the following structure:
\begin{framed}
\[\min\left\{\sum_{i\in N}\sum_{j\in E}\pi_{ij}(\vec u) x_{ij} \middle\vert \ell(\vec x)\leq \vec u, \vec x\in X\right\},\]
where we assume linear resource consumption $g_i(\vec x_i)=\vec x_i$ for all $\vec x_i\in X_i, i\in N$.
\end{framed}

\paragraph{Homogeneous Cost Functions.}
We first assume that cost functions 
are homogeneous, thus, 
the private cost of player $i\in N$ has the form
$ \pi_i(\ell(\vec x),\vec x_i):=\sum_{j\in E}c_{j}(\ell_j(\vec x))x_{ij}.$
We use in the following the more general model of
so-called \emph{polytopal
congestion games} introduced by
Del Pia et al.~\cite{PiaFM17} and further studied by Kleer and Sch\"afer~\cite{KleerS17}. In this model, 
 the strategy spaces are defined as
\[ X_i := P_i\cap \{0,1\}^m, i\in N,\]
where $P_i$ are polyhedrons of the form
$P_i=\{\vec x_i\in \R_+^m\vert A \vec x_i\geq \vec b_i\}$
for some rational matrix $A$ and integral vector $\vec b_i$
of appropriate dimension.
We remark here that all characterizations regarding box-TDI
and IPD also work for systems  $ A \vec x_i= \vec b_i$
or $A \vec x_i\leq \vec b_i$ assuming that $A$ and  $\vec b_i$
carry the desired structure (see for instance Kleer and Sch\"afer~\cite[Prop. 2.1]{KleerS17}). 
For homogeneous cost functions and polytopal strategy spaces,
we can use an LP formulation of Problem~\ref{price-opt} using the aggregation polyhedron $P_N$ as defined in Section~\ref{sec:aggregation}. Thus,  Theorem~\ref{thm:box-integral} implies the following result.
\begin{corollary}\label{cor:consequences}
Let $(N,X,(\pi_i)_{i\in N})$ be a congestion game with homogeneous nondecreasing
cost functions and polytopal strategy spaces with aggregation polyhedron
$P_N$. Let $\vec u$ be minimal for  $X$.
If $P_N$ is box-TDI and satisfies IDP, $\vec u$
is enforceable.
In particular, box-TDI and IDP holds for:
\begin{enumerate}
\item Network games with a common source and multiple sinks,
\item $r$-arborescence congestion games (see Harks et al.~\cite{Harks13mp} and Kleer and Sch\"afer~\cite{KleerS17} for a definition)
\item Intersection of strongly base-orderable matroids (see Kleer and Sch\"afer~\cite{KleerS17} for a definition)
\item symmetric totally-unimodular games (see Del Pia et al.~\cite{PiaFM17}
for a definition) including matching games,
\item Asymmetric matroid games (see Ackermann et al.~\cite{Ackermann08}).
\end{enumerate}
 \end{corollary}
 Note that by box-integrality and IDP of $P_N$,
any $\vec u\in \Z^m_+$ that minimizes
a strictly component-wise monotonically increasing
function is enforceable. In particular, for monotonically increasing functions $c_j(\ell_j(\vec x)), j\in E$, a vector $\vec u\in \Z^m_+$ 
corresponding to a minimum cost solution, that is, $\vec u=\ell(\vec x^*)$ for some
\[ \vec x^*\in \arg\min\left\{\sum_{j\in N} c_j(\ell_j(\vec x))\ell_j(\vec x)\middle\vert \vec x\in X\right\}\]
is enforceable. The congestion vector $\vec u$
minimizing the (weakly convex) social cost  can be computed in polynomial
time  (see Del Pia et al.~\cite{PiaFM17} and Kleer and Sch\"afer~\cite{KleerS17}) and additionally the space of enforcing
prices can be described by a compact linear formulation.
This allows for optimizing arbitrary linear objective functions 
(like maximum or minimum revenue) over the price/allocation space.
To the best of our knowledge, the only
previous results for the existence of (optimal) tolls
are due to Marden et al.~\cite{Marden09}, Fotakis and Spirakis~\cite{FotakisS08} and Fotakis et al.~\cite{FotakisKK10}.
Marden et al.~\cite{Marden09} proved that marginal cost tolls enforce
the minimum cost solution by charing the difference between
the social cost and the cost of Rosenthals' potential.
With this approach, there is no control on the magnitude of
price and no structure for optimizing secondary objectives over prices.
In addition, for other (non-optimal) vectors $\vec u$, this approach does not work.
Fotakis and Spirakis~\cite{FotakisS08} proved that any acyclic integral flow
in an $s$,$t$ digraph can be enforced.
It is not hard to see that the notion of minimality of $\vec u$
for $X$ exactly corresponds to the set of acyclic integral $s$,$t$ flows.
Fotakis et al.~\cite{FotakisKK10} generalized this result
to  single source multi-sink network games
allowing even for heterogeneous players.
\iffalse
For an $s_i,t_i$ path $P$
the combined cost for a heterogeneous player $i$ is defined as
\[ c_i(\vec x)=\sum_{j\in P} \alpha_i c_{j}(\ell_j(\vec x))\cdot x_{ij}+\lambda_j
\text{ for some }\alpha_i>0.\]
We note that the proof of Fotakis et al.~\cite{FotakisKK10}
seems incomplete because they 
only show that any edge load vector $\ell(\vec x^*)$
induced by an arbitrary acyclic integral flow $\vec x^*$ can be enforced as
a Wardrop equilibrium.  This property, however,  seems not
enough since the equilibrium definition in the atomic
game relies on a path decomposition which is not fully elaborated.

\fi

\iffalse
 (see also 
and Caragiannis et al.~\cite{CaragiannisKK10}).

They showed that network games with a single source
and heterogeneous users admit tolls
enforcing an cost minimal solution.
We briefly recall the setting heterogeneous users:
For an $s_i,t_i$ path $P$
the combined cost of player $i$ is defined as
\[ c_i(\vec x)=\sum_{j\in P} \alpha_i c_{j}(\ell_j(\vec x))\cdot x_{ij}+\lambda_j.\]
Clearly, this appears as a special case of the general formulation.

\fi
 
 \paragraph{Congestion Games with Player-Specific Cost Functions.}
 Now we turn to the general model
 of player-specific non-decreasing separable cost functions $c_{ij}(\ell_j(\vec x)), i\in N, j\in E$
 and consider the case of \emph{integral polymatroid congestion games},
 see Harks et al.~\cite{HarksKP18}.
In this model, for every $i\in N$, the strategy space $X_i$ is the integral base polyhedron $\B_{f_i}\subset \{0,1\}^m$ of a polymatroid $\P_{f_i}$.
Theorem~\ref{thm:mag} together with Theorem~\ref{thm:polymatroid-main} imply the following result.
  \begin{corollary}\label{cor:congestion-polymatroid}
  Consider an integral polymatroid congestion game with  
  $X_i=\B_{f_i}\subset \{0,1\}^m, i\in N$
  and nondecreasing player-specific separable cost functions.
  Let  $\vec u\in \Z^m$ be minimal for $X$.  Then, $\vec u$ is enforceable.
 \end{corollary}
 To the best of our knowledge, this is the first
 existence result of enforceable tolls in
 congestion games with player-specific cost functions.
 
One might be tempted to think that - as in Corollary~\ref{cor:consequences} -
one can use the aggregation polytope $P_N$ and
also show integrality of optimal solutions to~$LP^{\min}(\vec u)$ for 
natural classes such as $s$,$t$ network games with player-specific cost functions. Note that this approach does not work,  because the objective function depends on the
player's identities on the elements and not on an aggregate.
One can show that, unless $P=NP$, 
even for $s$,$t$ network games with homogeneous cost functions
but heterogeneous players, $LP^{\min}(\vec u)$ is not integral in general.
The proof is a straightforward reduction from the disjoint
path problem and omitted here.
\section{Application to Market Equilibria}\label{sec:market-equilibria}
We now consider market games and first present a classical model of Walrasian market equilibria with indivisible items. Then, we study a class of valuations for multi-item settings
that allows for some degree of externalities of allocations.

\subsection{Linear Pricing without Externalities}\label{sec:markets}
We are given a finite set  $E=\{1,\dots,m\}$ of \emph{items}
%and every item may be available at a certain multiplicity.
and there is a finite set of players $N=\{1,\dots,n\}$
interested in buying some of the items.
For every subset $S\subseteq E$ of items, 
player $i$ derives value $w_i(S)\in \R$
giving rise to a \emph{valuation function}
$w_i:2^m\rightarrow \R ,i\in N,$
where $2^m$ represents the set of all subsets of $E$.
The seller wants to determine
a price vector $ \bm\lambda\in \R_+^m$ so that
%all items are sold to the players and
 every 
player $i\in N$ gets a subset $S_i$ of items that are maximizers of her 
quasi-linear utility, that is, $S_i\in \arg\max_{S\subseteq E}\{ w_i(S)-\sum_{j\in S} \lambda_j \}$ and every item is sold to at most one player.
Such a tuple $((S_i)_{i\in N},\bm\lambda)$ is known as a \emph{competitive Walrasian equilibrium}.
A frequently assumed condition on valuations is normalization and monotonicity
as stated below.
\begin{enumerate}
\item $w_i(\emptyset) =0 $ for all $i\in N.$
\item $w_i(S) \leq w_i(T)$ for any $S\subseteq T\subseteq E$ and all $i\in N.$
\end{enumerate}

%\paragraph{Characterizing the Existence of Market Equilibria.}
We can derive an equivalent  game $G^{\max}(\vec 1)$ as follows.
Let $X_i= \{0,1\}^m, i\in N$ be the set of incidence vectors of the set $E$. 
The valuation
function is given by $v_i:\R^m\times X_i\rightarrow \R, (\vec 1,\vec x_i)\mapsto w_i(E(\vec x_i)),$ where $E(\vec x_i):=\{j\in E: x_{ij}=1\}$.
The resource consumption function are linear $g_i(\vec x_i)=\vec x_i, i\in N$.
We get the following characterization.
\begin{lemma}
The tuple  $((S_i)_{i\in N},\bm\lambda)$ is  a competitive Walrasian equilibrium   if and only if
the tuple $(\vec x^*,\bm\lambda)$ with $E(\vec x_i)=S_i, i\in N$ weakly enforces $\vec 1$ with market prices $\bm\lambda$ for $G^{\max}(\vec 1)$.
\end{lemma}
\begin{proof}
Note that the condition $\ell(\vec x^*)\leq 1$
ensures that every item goes to at most one player.
Thus, we get
\begin{align}\notag
&\hspace{0.5cm} \vec x_i^*  \in\arg\max\{v_i(\vec 1, \vec x_i) - \bm\lambda^\intercal \vec x_i\vert  \vec x_i \in X_i\} \\
\tag{By definition of $v_i$}\Leftrightarrow &\hspace{0.5cm}
\vec x_i^*  \in\arg\max\left\{w_{i}(E(\vec x_i)) - \bm\lambda^\intercal \vec x_i\middle\vert  \vec x_i\in X_i\right\}\\\notag
\Leftrightarrow & \hspace{0.5cm}
E(\vec x^*_i) \in\arg\max\left\{w_{i}(E(\vec x_i)) - 
\sum_{e\in E(\vec x_i)}\lambda_e
\middle\vert E(\vec x_i)\subseteq E\right\}\\ \notag
\Leftrightarrow & \hspace{0.5cm}
S_i \in\arg\max\left\{w_{i}(S_i) - 
\sum_{e\in S_i}\lambda_e
\middle\vert S_i\subseteq E\right\}.
\end{align}
\end{proof}
With this analogy to $G^{\max}(\vec 1)$ we can analyze Problem~\ref{price-opt-max} in more detail:
\[ \max\left\{\sum_{i\in N} v_i(\vec 1,\vec x_i) \middle\vert \vec x_i\in \{0,1\}^m, i\in N,\; \ell(\vec x)\leq \vec 1\right\}.\]
Clearly,  $X_i, i\in N$ consists of finitely many ($k_i=2^m$)
points and thus 
we can apply Theorem~\ref{thm:convex-hull-finite} to obtain a full characterization
of the existence of Walras market equilibria (which leads
precisely to the characterization of Bikchandani and Mamer~\cite{Bikchandani1997}).
\begin{corollary}[Bikchandani and Mamer~\cite{Bikchandani1997}]
Competitive Walrasian equilibria 
exist  if and only if the following LP admits integral optimal solutions:
\begin{framed}
\begin{align}\tag{LP$^{\max}$($\vec 1$)}
\max  \sum_{i\in N} \vec v_{i}^\intercal \bm \alpha_i,\;
 \ell(\bm\alpha)\leq \vec 1,\; \bm\alpha_i\in\Lambda_i \text{ for all } i\in N,
  %\sum_{h\in [k_i]} \alpha_{ih}&= 1, \text{ for all }i\in N \\
%\bm \alpha_i & \geq 0\text{ for all } i\in N.
\end{align}
where $\vec v_{i}:=(v_i(\vec 1,\vec x_i))_{\vec x_i\in X_i}$ and
$\ell(\bm \alpha):=\sum_{i\in N}  \sum_{j\in \{1,\dots,k_i\}} \alpha_{ij}$.
\end{framed}
\end{corollary}

%\paragraph{Gross-Substitute Valuations.}
A fundamental property of valuations $w_i, i\in N$
is the so-called \emph{gross-substitutes (GS) condition},
 requiring that whenever the prices of some items increase and the prices of other items remain constant, the agent's optimal demand for the items whose price remain constant only  increases.
Let us recall an existence theorem by Kelso and Crawford~\cite{Kelso82}.
\iffalse
defined below. 
Denote by $D_i(\bm \lambda)=\arg\max_{S\subseteq E}\{ w_i(S)-\sum_{j\in S} \lambda_j \}$ the set of maximizers for $i\in N$ given price vector $\bm\lambda$.
\begin{definition}
A valuation $w_i$ defined on $E$ satisfies the \emph{gross substitutes (GS)} condition if and only if 
for every price vector $\bm\lambda$, every set $S\in D_i(\bm \lambda)$, and every other price vector $\bm \mu\geq \bm \lambda$, there is a set $T\subseteq E$ with
$(S \setminus U)\cup T \in D_i(\bm \mu)$,
where $U := \{j\in E : \mu_j > \lambda_j\}$ is the of items whose prices have increased under $\bm \mu$ compared to $\bm \lambda$.
\end{definition}
\fi

\begin{theorem}[Kelso and Crawford~\cite{Kelso82}]\label{thm:kelso}
For GS valuations, there exists a competitive Walrasian equilibrium.\footnote{Gul and Stachetti~\cite{Gul99} even showed that in some sense
GS is the maximal condition on valuations for which equilibria exist.}
\end{theorem}

We can use Theorem~\ref{thm:main-max}
 by showing that problem~\ref{price-opt-max}
 has zero duality gap for the supply vector $\vec u=(1,\dots,1)^\intercal \in \R^m$.
 One can show this property by using insights from discrete
convexity and the special form of $M^{\natural}$-concave functions, see Murota~\cite[Sec. 11.3]{Murota:2003} for a definition and an exhaustive overview of the topic.
Let us now recall a characterization of Fujishige and Yang~\cite{FujishigeY03}.
\begin{theorem}[Fujishige and Yang~\cite{FujishigeY03}]
A normalized and monotone valuation function is GS if and only if it is $M^{\natural}$-concave.
\end{theorem} 
It is known that if all $v_i, i\in N$ are $M^{\natural}$-concave,
so is $\sum_{i\in N}v_i$. Altogether, problem~\ref{price-opt-max}
is a very special discrete optimization problem
with an $M^{\natural}$-concave function over $\{0,1\}^{n\cdot m}$
involving the special constraint $\ell(\vec x)\leq \vec 1$ which
constitues a \emph{laminar system} or \emph{hierarchy}, see Yokote~\cite{Yokote2018} and Budish et al.~\cite{Budish13}
for further details.
Yokote~\cite{Yokote2018} proved that
such a  problem has zero duality gap.
Thus, with the zero duality gap property and the monotonicity
of $v(\vec x)$, Theorem~\ref{thm:main-max} implies  Theorem~\ref{thm:kelso}.
We remark here that Yokote~\cite[Sec.~4]{Yokote2018}
described in his paper the connection
of his strong duality theorem (in the realm of discrete convexity) with the existence of market equilibria for GS valuations.

\subsection{Nonlinear Package Pricing}
The auction model so far assumes a single seller that uses
linear anonymous pricing functions, that is,
every item $j$ comes with a price $\lambda_j\geq 0$
and the price of every subset of items $S$ is linear in item prices, that is,  $\lambda(S)=\sum_{j\in S}\lambda_j$.  Bikhchandani and Ostroy~\cite{BikhchandaniO02} studied a model  in which \emph{packages of items} are sold and the sellers may use
nonlinear (non-anonymous) pricing functions (see also Parkes and Ungar~\cite{ParkesU00}).  
Pricing functions assign prices to packages (instead of prices for individual items) that may depend on the package type 
only, or on the package and buyer (or seller) identity,
or on the identity of all parties, i.e., the package, buyer and seller  (see  Bikhchandani and Ostroy~\cite{BikhchandaniO02}). 
\iffalse
There is a set of items $I=\{1,\dots, k\}$ and
a player set $N$ that can be
partitioned into a set of sellers $A$ and buyers $B$. 
Each seller $s\in A$ can offer a set of packages $z_{sj}\subset I, j\in A_s$ of items for sale and each buyer $b$ buys a set of packages from the sellers. Buyers (sellers) have valuation functions $v_b$ (or $v_s$) defined over the bought (sold) packages, respectively.
Nonlinear pricing functions assign prices to packages (instead of prices for individual items) that may depend on the package type 
only, or on the package and buyer (or seller) identity,
or on the identity of all parties, i.e., the package, buyer and seller  (see  Bikhchandani and Ostroy~\cite{BikhchandaniO02}). 
\fi
A competitive equilibrium arises, if there are package prices and package allocations
so that the the allocation maximizes the  overall quasi-linear utility of every buyer and seller, respectively.

One can incorporate this model into the current framework by defining
an appropriate game $G^{\max}(\vec 0)$ as follows. 
The resource set $E$ is constructed according
to the qualitatively different packages traded on the market.
Packages of the same type correspond to a 
resource (yielding a price function per type)
but   packages with 
dependencies on the buyers/sellers
would correspond to individual resources.

Buyers $b\in B$ have as strategy space $X_b\subseteq \Z_+^m$, where  $x_{bj}$ is the number of packages of type $j$  (or from seller $j$ if the
seller identity matters)
while the strategy space  of every
seller $s\in A$ is given by some set  $X_s\subseteq\Z_+^p$, where  $x_{sj}$ represents the number of packages of type $j$ produced by seller $s$.
By assigning
resource consumption functions $g_b(\vec x_b)=\vec x_b\in \Z^m_+, b\in B$
for buyers and $g_s(\vec x_s)=-\vec x_s\in \Z^m_-, s\in A$
for sellers, respectively,
a competitive equilibrium then corresponds
to a pair $(\vec x^*,\bm\lambda)\in X\times \R_+^m$
that weakly enforces $0$ for $G^{\max}(\vec 0)$ with market prices.
Note that the condition $\ell(\vec x)\leq 0$
ensures that supply exceeds demand and the dual variable $\lambda$ corresponds to the
market clearing equilibrium prices (which may be nonlinear in terms
of item prices).
In this construction, the level of price differentiation
depends the constructed set $E$, that is, $E$ might model
the number of anonymous package types (leading to anonymous
package prices), or packages prices depending on 
the buyer and seller (leading to an increased number $|E|$).
 The LP characterization of Bikhchandani and Ostroy~\cite{BikhchandaniO02} regarding the existence
of competitive equilibria can be deduced from Theorem~\ref{thm:convex-hull-finite}, because the game $G^{\max}(\vec 0)$  exhibits a finite
strategy space for every player (buyer or seller) and besides the demand-supply condition the utilities are separable over players.
  In fact, looking at Theorem~\ref{thm:convex-hull-finite}, one can generalize the characterization of Bikhchandani and Ostroy~\cite{BikhchandaniO02} along several directions.
One direction is to allow that players may be buyers and
sellers at the same time. 
In Section~\ref{sec:trading}, we consider so-called trading networks
that exhibit this property and the reduction we present
is very similar to the one sketched here.

\subsection{Item Multiplicity, Additive Linear Valuations with Externalities and Polymatroids}
Now we turn to a multi-item model that allows for several items
of the same type and some degree of externalities of allocations. 
There is a finite set  $E=\{1,\dots,m\}$ of \emph{item types}
and every item may be available at a certain multiplicity.
Assume further that $X_i\subset \{0,1\}^m, i\in N$.  This implies that every player
wants to receive at most one item per type - however $X_i$
may still carry some combinatorial restrictions for feasible
item sets for player $i\in N$.
Suppose that valuations of players
are additive  over items, that is, 
\begin{equation}\label{val:add} v_i(\ell(\vec x),\vec x_i):= \sum_{j \in E} v_{ij}(\ell_j(\vec x)) x_{ij},\end{equation}
where $v_{ij}:\Z_+\rightarrow \R_+$ is the nonnegative value player $i$ gets
from receiving  an item of type $j$ assuming that item $j$
is sold to $\ell_j(\vec x)$ many players.
This formulation is not directly comparable
to the one before. On the one hand side,
additivity of valuations over items
is less general. On the other hand, several items of the same type can be sold and we allow for a functional
dependency of the valuations with respect to the load $\ell_j(\vec x)$. Such dependency may be interesting
for situations, where the value $v_{ij}(\cdot)$ of receiving item type $j$ drops
as other players also receive the same type -- 
this is referred to as a \emph{setting with negative externalities}.
\begin{assumption}\label{ass:value-seprable}
For every $i\in N, j\in E$, the functions $v_{ij}$
are nonnegative and
exhibit \emph{negative externalities}, that is, 
$v_{ij}(z)\geq v_{ij}(z+1) \text{ for all } z\in \Z_+.$
\end{assumption}
The model so far does not satisfy the assumption of a maximization
game $G^{\max}(\vec u)$ augmented with prices, as
the utility of a player is allowed
to depend on the load vector $\ell(\vec x)$.
However, one can easily verify that with Assumption~\ref{ass:value-seprable},
the  game is in fact a monotone aggregative
game as defined in Section~\ref{sec:monotone}.
With this insight at hand,
by Theorem~\ref{thm:mag} it suffices 
to analyze cases for which
$\vec u$ is enforceable for $G^{\max}(\vec u)$,
which in turn is equivalent to the property
that Problem~\ref{price-opt} has zero duality gap and 
and satisfies $\ell(\vec x^*)=\vec u$ for a primal-dual optimal pair $(\vec x^*\bm\lambda)$.
We get the following result using Theorem~\ref{thm:max-polymatroid}.
\begin{corollary}\label{cor:market-polymatroid}
Let $X_i=\P_{f_i}\subset \{0,1\}^m, i\in N$ be integral polymatroid polyhedra
and assume that valuation functions satisfy 
Assumption~\ref{ass:value-seprable}.
Then, any supply vector $\vec u\in \Z^m_+$ for which there is $\vec x\in X$ with $\ell(\vec x)=\vec u$ is enforceable.
\end{corollary}
\begin{proof}
By Assumption~\ref{ass:value-seprable} and the structure
of the valuation functions (see~\eqref{val:add})
it follows directly that the game is a monotone aggregative
game. Thus, for any supply vector $\vec u\in \Z^m_+$ with $\ell(\vec x)=\vec u$ for some $\vec x\in X$, by  Theorem~\ref{thm:mag} it suffices 
to show that $\vec u$ is enforceable for $G^{\max}(\vec u)$. 
This, however, follows directly from Theorem~\ref{thm:max-polymatroid}.
\end{proof}

\section{Application to Trading Networks}\label{sec:trading}
A bilateral trading network is represented by a directed multigraph $G = (N, E)$, where $N$ is the set of vertices and $E=\{e_1,\dots,e_m\}$ the set of edges. Each vertex corresponds to a player and each edge $e=(s,b)$ represents a bilateral trade that can take place between the  pair of incident vertices $s,b\in N$. For each $e=(s,b)\in E$, the source vertex $s$ corresponds to the seller and the sink vertex $b$ corresponds to the buyer in the trade. For $i\in N$, let $\delta^+(i)$ and $\delta^{-}(i)$ be the set
of outgoing and incoming edges of vertex $i\in N$
and as usual we denote the set of all edges incident  to $i$
by $\delta(i)=\delta^+(i)\cup \delta^-(i)$. 
For a set of edge prices $\lambda_e\geq 0, e\in E$, we can associate
with each possible trade $e=(s,b)\in E$ 
a price $\lambda_e\geq 0$ with the understanding
that the buyer $b$ pays $\lambda_e$ to the seller $s$.
An outcome of the market is a set of 
\emph{realized trades} $S\subseteq E$ and a vector of prices $\bm\lambda\in \R_+^m$.
Given an outcome, the quasi-linear utility of a player $i \in N$  is defined as
the  
sum of the utility gained from trades plus the income
minus the cost of trades, respectively.
The utility of
realized trades is given by a function 
$\bar w_i:2^{\delta(i)}\rightarrow \R.$
We  extend $\bar w_i$ to $2^m$ by taking
$ w_i:2^{m}\rightarrow \R,  S\mapsto \bar w_i(S\cap\delta(i)).$
The overall utility for given $S\subseteq E$ and $\bm\lambda\in \R_+^m$
is defined as
\begin{equation}\label{eq:trade-utility}
w_i(S)+\sum_{e\in \delta^+(i)\cap S}\lambda_e-
\sum_{e\in \delta^-(i)\cap S}\lambda_e
\end{equation}
For the function $w_i,i\in N$
we only assume monotonicity on the buyer side, that
is, $w_i(S)\geq w_i(T)$ for all $T\subseteq S\subseteq \delta^-(i)$.
Free disposal at the buyer side is a sufficient condition for this assumption.
The market maker wants to determine a price vector $\bm\lambda\in \R_+^m$ and a set of realized trades $S^*\subseteq E$ such that
\[S^*  \in \arg\max_{S\subseteq E}\left\{w_i(S)+\sum_{e\in \delta^+(i)\cap S}\lambda_e-
\sum_{e\in \delta^-(i)\cap S}\lambda_e\right\} \text{ holds for all }i\in N.\]
Such a tuple $(S^*,\bm\lambda)$ constitutes a competitive equilibrium.
 The main difference to the Walrasian market equilibrium
model is that players
 can simultaneously act as buyers and sellers in different trades.
 
We will cast this problem in the framework by constructing
an equivalent game $G^{\max}(\vec 0)$.
For each player $i$, we have a vector $\vec x_i\in \{-1,0,1\}^m$
with the understanding that 
\[ x_{ie}=\begin{cases}-1,& \text{ if }e\in \delta^+(i) \text{ and trade $e$ is realized as seller}\\
1, &\text{ if } e\in \delta^-(i) \text{ and trade $e$ is realized as buyer}\\
0, &\text{ if } \text{$e\notin \delta(i)$ or $e$ is not realized.}
\end{cases}\]
We thus define $X_i= \{ \vec x_i\in \{-1,0,1\}^{m}\vert x_{ie}=0, e\notin \delta(i)\}, i\in N$.
To complete the description of $G^{\max}(\vec 0)$, we assume
that the resource consumption functions are given as $g_i(\vec x_i)=\vec x_i$ for all $i\in N$ and we define the valuation function of player $i\in N$ on $X_i$ by
 \begin{equation}\label{eq:valuation-trade} v_i(\vec 0, \vec x_i) := w_{i}(\{e\in \delta(i):\; |x_{ie}|=1\}) .
 \end{equation}
 With this construction, we have a one-to-one
 correspondence between $\vec x_i\in X_i$
 and sets $S_i\subseteq \delta(i)$ via
 $E(\vec x_i):=\{e\in \delta(i):\; |x_{ie}|=1\}$.
We obtain the following characterizations
on the existence of competitive equilibria
using the notation $E(\vec x):=\cup_{i\in N} E(\vec x_i)$.
\begin{lemma}
Consider a bilateral trading game and
let $G^{\max}(\vec 0)$ be an associated pricing game.
Then, the following statements are equivalent.
\begin{enumerate}
\item\label{enum:trade1} There exists a competitive
equilibrium $(E(\vec x^*),\bm\lambda)\in E\times \R_+^m$ for the bilateral trading game.
\item\label{enum:trade2} The vector $\vec u=\vec 0$ is enforceable via
$(\vec x^*,\bm\lambda)\in X\times \R_+^m$ for the game $G^{\max}(\vec 0)$.
\item\label{enum:trade3} $P^{\max}(\vec 0)$  has zero duality gap and $\vec x^*\in X$ is 
 an optimal solution $\vec x^*\in X$ that satisfies $\ell(\vec x^*)=0$.
\end{enumerate}
\end{lemma}
\begin{proof}
By Theorem~\ref{thm:main-max}
we have already that \eqref{enum:trade2}$\Leftrightarrow$\eqref{enum:trade3} holds, so we only need to show
\eqref{enum:trade1}$\Leftrightarrow$\eqref{enum:trade2}.
For any $\vec x\in X$ 
 with $\ell(\vec x)=\vec 0$,
we have   \begin{equation}\label{eq:consistent} 
\forall e=(s,b)\in E:\;\; e=(s,b)\in E(\vec x_s) \Leftrightarrow e=(s,b)\in E(\vec x_b).\end{equation}
For $i\in N$ arbitrary, we get
\begin{align}\notag
&\hspace{0.5cm} \vec x_i^*  \in\arg\max\{v_i(\vec 0, \vec x_i) - \bm\lambda^\intercal \vec x_i\vert  \vec x_i \in X_i\} \\
\tag{By definition of $v_i$}\Leftrightarrow &\hspace{0.5cm}
\vec x_i^*  \in\arg\max\left\{w_{i}(E(\vec x_i)) - \bm\lambda^\intercal \vec x_i\middle\vert  \vec x_i\in X_i\right\}\\\notag
\Leftrightarrow & \hspace{0.5cm}
E(\vec x^*_i) \in\arg\max\left\{w_{i}(E(\vec x_i)) + 
\sum_{e\in E(\vec x_i)\cap \delta^+(i)}\lambda_e-
\sum_{e\in E(\vec x_i)\cap \delta^-(i)}\lambda_e
\middle\vert E(\vec x_i)\subseteq \delta(i)\right\}\\ \tag{By def. of $w_{i}$ and~\eqref{eq:consistent}}
  \Leftrightarrow &\hspace{0.5cm} E(\vec x^*) \in\arg\max\left\{w_{i}(S) +\sum_{e\in \delta^+(i)\cap S}\lambda_e-
\sum_{e\in \delta^-(i)\cap S}\lambda_e\middle\vert S\subseteq E\right\} \end{align}
\end{proof}
 With this characterization, we can use the results obtained so far
 for the enforceability of $\vec 0$
for $G^{\max}(\vec 0)$. Note that every $X_i,i\in N$
consists of  $k_i:=3^{|\delta(i)|}$ many points, thus, Theorem~\ref{thm:convex-hull-finite}
gives a complete characterization 
of the existence of competitive equilibria.
 \begin{corollary}\label{cor:trading-LP}
 %$(E(\vec x^*),\bm\lambda)\in X\times \R^m_+$
Competitive equilibria 
for bilateral trading networks exist if and only if the following LP admits integral optimal solutions
$\bm\alpha$ with $\ell(\bm \alpha)=\vec 0$:
\begin{framed}
\begin{align}\tag{LP$^{\max}$($\vec 0$)}\label{bilateral:LP}
\max  \sum_{i\in N} \vec v_{i}^\intercal \bm \alpha_i,\;
 \ell(\bm\alpha)\leq \vec 0,\; \bm\alpha_i\in\Lambda_i \text{ for all } i\in N,
  %\sum_{h\in [k_i]} \alpha_{ih}&= 1, \text{ for all }i\in N \\
%\bm \alpha_i & \geq 0\text{ for all } i\in N.
\end{align}
where $\vec v_{i}:=(v_i(\vec 0,\vec x_i))_{\vec x_i\in X_i}$ and
$\ell(\bm \alpha):=\sum_{i\in N}  \sum_{j\in \{1,\dots,k_i\}} \vec x_i^j\alpha_{ij}$.
\end{framed}
\end{corollary}

We obtain the following result as a direct corollary of Theorem~\ref{concave-main}.
\begin{corollary}
\ref{bilateral:LP} can be solved in polynomial time,
if there is a polynomial time demand oracle.
\end{corollary}

In order to obtain existence results, one needs
to enforce some assumptions on the valuation
functions $w_i, i\in N$. Hatfield et al.~\cite{Hatfield13}
introduced the concept of \emph{fully substitutable}
valuations. We omit here the precise definition
but it is important to know that this concept is in fact
equivalent to the known concept of GS
valuations or $M^\natural$-concave valuations (see Hatfield et al.~\cite{HatfieldKNOW15}). Instead of the monotonicity
property of valuations (as e.g. in Fujishige and Yang~\cite{FujishigeY03}), we assume that
valuations are \emph{buyer-monotone},
that is, for every $i\in N$ and $S\subset T \subseteq \delta^-(i)$, we have $w_i(S)\leq w_i(T)$. Free disposal for the buyer is a sufficient
condition.
We obtain the following result.
\begin{theorem}[Hatfield et al.~\cite{Hatfield13}]
For fully substitutable buyer-monotone valuations, there exists a competitive
equilibrium.
\end{theorem}
\begin{proof}
In order to apply our previous results,
we need to check whether
 \begin{equation}\label{master-bilateral} \max\left\{ v(\vec x):=\sum_{i\in N}v_i(\vec 0, \vec x_i)\middle\vert 
\vec x \in  X, \; \ell(\vec x)\leq \vec 0 \right\}\end{equation}
 has zero duality gap and admits an optimal solution
$\vec x^*\in X$ with $\ell(\vec x^*)=\vec 0$.
Again the result of Yokote~\cite{Yokote2018} implies zero duality gap
as $\ell(\vec x)\leq \vec 0$ is a laminar system.
With the buyer monotonicity, any optimal
solution $\vec x^*$ to~\eqref{master-bilateral} 
can be turned into one with $\ell(\vec x^*)= \vec 0$,
hence, Theorem~\ref{thm:main-max} implies  Theorem~\ref{thm:kelso}.
 \end{proof}

 \section{Application to Congestion Control in Communication Networks}\label{sec:congestion-control}
In the domain of network-based TCP congestion control, we are given a directed  \emph{capacitated} graph $G=(V,E,\vec u)$,
where $V$ are the nodes, $E$ with $|E|=m$ is the edge set and
$\vec u \in \R_+^m$ denote the edge capacities.
There is a set of players $N= \{1, \dots,
n\}$ and every $i \in N$ is associated with an end-to-end pair $(s_i,t_i)\in V\times V$ and a bandwidth utility function $U_i:\R_+\rightarrow\R_+$
measuring the received benefit from sending net flow from $s_i$ to $t_i$.
As in congestion games, a \emph{flow} for~$i\in N$ is a nonnegative vector
$\vec x_i \in \R^{|E|}_+$ that  lives in the flow polyhedron: 
\begin{align*}
X_i=\left\{\vec x_i\in \R_+^m\middle \vert \sum_{j\in \delta^+(v)} x_{ij} - \sum_{j\in \delta^-(v)} x_{ij} = 0, \text{ for all } v\in V\setminus\{s_i,t_i\}\right\},
\end{align*}
where $\delta^+(v)$ and $\delta^-(v)$ are the arcs leaving and
entering~$v$.
We assume $X_i\neq \emptyset$ for all $i\in N$ and
we denote the net flow reaching $t_i$
by $\val(\vec x_i):= \sum_{j\in \delta^+(s_i)} x_{ij} - \sum_{j\in \delta^-(s_i)} x_{ij}, i\in N$.
The goal in price-based congestion control is to determine
edge prices $\lambda_j, j\in E$ so that 
a strategy distribution $\vec x^*$ is induced as an equilibrium
respecting the network capacities $\vec u$ and, hence, avoiding
congestion. Assuming that resource consumption is linear,
that is, $g_i(\vec x_i)=\vec x_i, i\in N$, the equilibrium condition amounts to
\[ \vec x_i^*\in \arg\max\{U_i( \val(\vec x_i))-\bm\lambda^\intercal\vec x_i \vert \vec x_i\in X_i\} \text{ for all $i\in N$.}\]
We obtain the following result for concave bandwidth utility functions.
\begin{theorem}[Kelly et al.~\cite{Kelly98}]
For concave bandwidth utility functions $U_i, i\in N$, every capacity vector $\vec u\in \R^m_+$ is weakly enforceable with market prices.
\end{theorem}
\begin{proof}
With the concavity of  $U_i, i\in N$, problem~\ref{price-opt-max}
is a convex optimization problem over a polytope
and hence satisfies Slater's constraint qualification conditions
for strong duality. Thus, Theorem~\ref{thm:main-max} implies the result.
\end{proof}

Let us turn to models, where the flow polyeder $X_i$ is intersected with $\Z_+^m$.
Most of the previous works in the area of congestion control
assume either that there is only a single path per $(s_i,t_i)$ pair or as in
Kelly et al.~\cite{Kelly98}, the flow is allowed to be fractional. Allowing a fully fractional distribution of the flow, however, is not possible in some interesting applications - the notion of data packets as indivisible units 
seems more realistic.  The issue of completely fractional
routing versus integrality requirements has been explicitly addressed by Orda et al.~\cite{Orda93}, Harks and Klimm~\cite{HarksK16b} and Wang et al.~\cite{wang2011}.
Using the TDI and IDP property of network
matrices, we obtain the following result for integral
flow polytopes.
\begin{corollary}\label{cor:kelly-integral}
Let  the bandwidth utility functions $U_i, i\in N$ be non-decreasing, identical and linear and assume that all players share the same source $s_i=s, i\in N$.
Then, for integral routing models with strategy spaces  $X'_i=X_i\cap\Z_+^m$,  every capacity vector $\vec u\in \Z^m_+$ is weakly enforceable with market prices.
\end{corollary}
\begin{proof}
For problem~\ref{price-opt-max},
we can w.l.o.g. change the instance
by introducing a super-sink
and connect all $t_i$'s to the sink with large enough integral capacity. 
This way, we obtain an ordinary  $s$-$t$
max-flow problem for which the LP-formulation
 $LP^{\max}(\vec u)$ is known to be integral.
\iffalse
We can write
$X'_i=\{\vec x_i\in \Z_+^m\vert A\vec x_i=\vec b_i\}$ for all $i\in N$, where $A$ is the graph incidence matrix of $G$ and
$b_{ij}=0$ for all $j\in V\setminus\{s_i, t_i\}$ while for $s_i$,$t_i$
there are no constraints.
Then, we can use the aggregation polytope $P_N$ as in Section~\ref{sec:atomic-congestion-games}, equation~\eqref{eq:aggregation-polytope}.
The assumption on bandwidth utility functions implies the form $U_i(z)=a z, a\geq 0, i\in N$.  Thus, the aggregated utility can be written as
\[ \sum_{i\in N}\val(\vec x_i)=  \sum_{i\in N}a\left(\sum_{j\in \delta^+(s)} x_{ij} - \sum_{j\in \delta^-(s)} x_{ij}\right)= a \left(\sum_{j\in \delta^+(s)} y_{j} - \sum_{j\in \delta^-(s)} y_{j}\right).\]
Then, $LP^{\max}(\vec u)$ can be reformulated as
 \[ \max\left\{a \left(\sum_{j\in \delta^+(s)} y_{j} - \sum_{j\in \delta^-(s)} y_{j}\right) \;\middle\vert\;  \vec y\in P_N\cap\{\vec y \vert \vec y\leq \vec u\} \right\}\]
As the objective is linear and $P_N$ is box-TDI, this LP admits
an integral optimal solution. By the IDP property of $P_N$, we can decompose an integral
optimal solution and the result follows.
One can also interpret $LP^{\max}(\vec u)$
as a max-flow problem on a slightly changed
instance, where we introduce a super-sink
and connect all $t_i$'s to the sink with large enough capacity.
This way, we obtain a standard max-flow problem
which is known to admit integral optimal solutions.
\fi
\end{proof}
\begin{remark}
The above proof shows that for a capacity vector $\vec u\in \Z^m_+$, we can compactly represent the enforcing prices/allocation space
and efficiently optimize linear functions over it.
\end{remark}
While the above result seems to require somewhat
restrictive assumptions (linear identical bandwidth utilities
and a common source), we show in the following that
already for two source-sink pairs with
identical  linear capped bandwidth utilities,
enforceability is not guaranteed, unless $P= NP$.
A capped linear function $f:\R\rightarrow \R$
has the form $f(x)=ax$, for $x\leq x^{\max}$
and  $f(x)=a x^{\max}$ for  $x\geq x^{\max}$.
This type of function is concave and arises quite
naturally as we only require the existence of an
upper bound on the requested bandwidth of every player.
\begin{proposition}\label{prop:congestion-control-reduction}
Unless $P= NP$, there is an instance with only two players with different source sink pairs $(s_i,t_i), i\in\{1,2\}$ and non-decreasing, identical and linear
capped bandwidth utilities $U_i, i\in\{1,2\}$, for which there is a vector $\vec u\in \Z^m_+$
that is not weakly enforceable by market prices.
\end{proposition}
\begin{proof}
Having capped bandwidth utilities implies that
there are only finitely many strategies per player.
Thus,  we can use the LP characterization result
of Theorem~\ref{concave-main}:
It remains to prove that the master problem~\ref{price-opt-max}
is NP-hard and that the demand problem is
polynomial time solvable.
The demand problem amounts to
\[ \max\{\val(x_i)-\bm\lambda^\intercal x_i\vert \vec x_i\in X_i\},\]
which is just a max flow problem.
For the master problem, it is not hard
to see that we can reduce from the two-directed
disjoint path problem. For an instance of two-directed
disjoint path, we associate the given two source-sink pairs
naturally with those of two players $\{1,2\}$ and assume $\vec u=1$ and $U(\val(x_i))=\val(x_i), i\in \{1,2\}$
with a cap at any value larger equal $1$.
This way, due to the integrality of the flows, there is a solution
to the disjoint path problem iff the objective value of the
master problem is $2$.
\end{proof}

\section{Conclusions and Extensions}
We introduced a generic resource allocation
problem and studied the question of enforceability
of certain load vectors $\vec u$ via (anonymous) 
pricing of resources. We derived a characterization
of enforceable load vectors 
via studying the duality gap of an associated optimization problem.
We further derived a characterization connecting enforceability
for arbitrary non-convex settings to enforceability of a convex
model. 
Using this general result, we studied consequences
of known structural results in the area of linear integer
optimization, polyhedral combinatorics and
discrete convexity for several application cases.

Understanding duality gaps of optimization problems is an active research area, see for instance the progress on duality for
nonlinear mixed integer programming (cf. Baes et al.~\cite{BaesOW16}).
Thus, our general characterization yields the opportunity
to translate progress in this field to economic situations
mentioned in the applications.

\iffalse
A further consequence of the proposed framework is the enforceability
of load vectors $\vec u$ using \emph{mixed}
or \emph{correlated} equilibria. For these equilibrium concepts, the
strategy space of a finite strategic game is a (convex) polytope and
if the cost/utility function of the extended game (e.g., the 
cost/utility function of the mixed extension)
is convex in the randomization variable, we have strong duality of the master problem~\ref{price-opt}
and enforceability results for expected load vectors follow. 
\fi

For our general model we assumed that the strategy spaces
are subsets $X_i\in \R^m, i\in N$. This assumption
is not necessary for proving our main result.
We could have chosen $X_i$ as a Banach space
and the results would have gone through.
In fact, in the area of dynamic traffic assignments (cf. Friesz et al~\cite{Friesz93}), the flow trajectories live in 
function spaces, thus, offering the possibility
that our characterization on Banach spaces
yields the existence of (time varying) tolls 
for these applications too.

\subsection*{Acknowledgements}
I thank Dimitris Fotakis, Lukas Graf, Martin Hoefer, Max Klimm, Anja Schedel and Julian Schwarz for
helpful discussions and comments on an earlier draft of this manuscript.
I am also grateful for the  comments received by attendees of the workshop
``20 years of price of anarchy'' held in Crete, Juli 2019.

\bibliographystyle{plain}
\bibliography{../master-bib}

% Appendix
\appendix

\section{Utility Maximization Problems}\label{subsec:max}
We turn to utility maximization problems and define the following
analogous problem:
\begin{framed}
\begin{equation}\tag{$P^{\max}(\vec u)$}\label{price-opt-max}
\begin{aligned}
\max\;\left\{ v(\vec x) \vert \;\;  \ell_j(\vec x) \leq u_j, \; j\in E,
\;\; \vec x_i\in X_i, \; i=1,\dots,n\right\},
\end{aligned}
\end{equation}
where the objective function is defined as
$v(\vec x):=\sum_{i\in N}v_i(\vec u,\vec x_i).$
\end{framed}
The Lagrangian function for problem~\ref{price-opt-max} becomes
$L(\vec x,\bm\lambda):=v(\vec x) -\bm\lambda^\intercal (\ell(\vec x)-\vec u) ,\; \bm \lambda\in \R_+^m,$
and we can define the Lagrangian-dual as:
\begin{align*} \mu : \R_+^m \rightarrow\R,\;\;
\mu(\bm \lambda)=\sup_{\vec x \in X} L(\vec x,\bm \lambda)=\sup_{ \vec x \in X}\{v(\vec x)-\bm\lambda^\intercal (\ell(\vec x)-\vec u)\}.
\end{align*}
We assume that $\mu(\bm \lambda)=\infty$, if $L(\vec x,\bm \lambda)$ 
is not bounded from above on $X$.
The \emph{dual problem} is defined as:
\begin{align}\label{price-dual-max}
\tag{$D^{\max}(\vec u)$} \inf_{\bm \lambda\geq 0} \mu(\bm \lambda)\end{align}
We obtain the following analogous results to the minimization case.
\begin{theorem}\label{thm:main-max}
Consider a game of type~$G^{\max}(\vec u)$.
Then, the following statements hold:
\begin{enumerate}
\item\label{enum:main1-max} A supply vector $\vec u\in \R^m$ is enforceable
via $(\vec x^*,\bm \lambda^*)$ if and only if $(\vec x^*,\bm \lambda^*)$ has zero duality gap for~\ref{price-opt-max} and $\vec x^*$ satisfies~\eqref{eq:inequality} with equality. 
\item\label{enum:main2-max} A supply vector $\vec u\in \R^m$ is weakly enforceable
with market clearing prices
if and only if $(\vec x^*,\bm \lambda^*)$ has zero duality gap for~\ref{price-opt-max}. 
\item\label{enum:main-max-unique} A supply vector $\vec u\in \R^m$ is uniquely enforceable via $(\vec x^*,\bm \lambda^*)$
if and only if $(\vec x^*,\bm \lambda^*)$ has zero duality gap for~\ref{price-opt-max} and $\vec x^*$ is
a unique optimal solution for~\ref{price-opt-max}
satisfying~\eqref{eq:inequality} with equality. 
\iffalse
\item\label{enum:main3-max} A supply vector $\vec u\in \R^m$ is weakly enforceable via $(\vec x^*,\bm \lambda^*)$,
if $(\vec x^*,\bm \lambda^*)$ has zero duality gap for~\ref{price-opt-max}.
\fi
\end{enumerate}
\end{theorem}
In maximization games, the sets $X_i$
usually contain some capacity restrictions,
therefore the notion of minimality of vectors $\vec u$
might not be appropriate. Take for instance the
example of auctions in Example~\ref{ex:auction}. Here,   $\vec u=0$ arises
as the unique minimal $\vec u$ leading
to trivial conclusions.
Perhaps
 more interesting are scenarios in which the combined valuation function $v(\vec x)$ is in some sense monotonically non-decreasing on $X$.
 
 \begin{definition}[Upwards closure of $X$, Monotonicity of valuations]
 $X$ is \emph{upwards-closed} w.r.t. $\vec u\in \R^m$, if $\ell(\vec x)\leq \vec u$ and $\ell(\vec x)\neq \vec u$ for some $\vec x\in X$ implies that
 there is $\vec x'\in X$ with $\vec x'\geq \vec x$ and
 $\ell(\vec x')= \vec u$. We say that  $X$ is \emph{upwards-closed},
 if this property holds for all $\vec u\in \R^m$.
  The function $v(\vec x)$ is \emph{monotonically non-decreasing} on $X$, if
 $v(\vec x)\geq v(\vec y)$ for all $\vec x, \vec y\in X$ with $\vec x\geq \vec y$.
 \end{definition}
We obtain the following result regarding this monotonicity assumption.
\begin{theorem}
Assume that $v(\vec x)$
is monotonically non-decreasing and $X$ is upwards closed w.r.t. $\vec u\in \R^m$.
Then, the supply vector $\vec u\in \R^m$ is enforceable via $(\vec x^*,\bm\lambda^*)$ if and only if $(\vec x^*,\bm\lambda^*) $ 
has zero duality gap for~\ref{price-opt-max}. 
\end{theorem}
\begin{proof}
By the monotonicity of $v(\vec x)$ and upwards-closedness of $X$
w.r.t. $\vec u$, any optimal solution of~\ref{price-opt-max} can be turned
into one that satisfies~\eqref{eq:inequality} with equality.
\end{proof}

We finally get an existence result for convex sets $X_i, i\in N$
and monotone and concave valuations.

\begin{corollary}
Let $X_i, i\in N$ be nonempty convex sets such that $X$ is upwards-closed
w.r.t. $\vec u\in \R^m$.
Assume that $v_i, i\in N$ are concave functions, $g_i, i\in N$ are convex functions
and $v(\vec x)$
is monotonically non-decreasing and that there exists
$\vec x^0\in \relint\left(\{\vec x\in X\vert \ell(\vec x)\leq \vec u\}\right) $.
Then, $\vec u$  is enforceable. If $v(\vec x)=\sum_{i\in N}v_i(\vec u, \vec x_i)$ is strictly concave 
over $X$, then $\vec u$ is uniquely enforceable.
\end{corollary}
\section{Proof of Theorem~\ref{thm:polymatroid-main}}
Let us now restate problem~$P^{\min}(\vec u)$ in the context
of polymatroids.

 \begin{framed}
\begin{align}\tag{$P^{\min-\polymatroid}(\vec u)$}\label{cp-polymatroid}
\min\;\sum_{i\in N}&\sum_{j\in E}  \pi_{ij}(\vec u) x_{ij}\\\notag
\vec x_i&\in \B_{f_i},\;i\in N\\
\ell_j(\vec x)&\leq u_j, j\in E\label{ineq:load-polymatroid}
%x_{ij}&\geq 0\; \forall i\in N, j\in E.
\end{align}
We call $LP^{\min-\polymatroid}(\vec u)$ the fractional
relaxation, where
we optimize over  $\EB_{f_i}, i\in N$
instead of $\B_{f_i}, i\in N$.
\end{framed}

\begin{proof}
We first lift all integral base polyhedra $\B_{f_i}\subset \Z^m$
to the higher dimensional space $\bar \B_{f_i}\subset \Z^{n\cdot m}$
by introducing $n$ copies $E_i, i\in N$ of the elements $E$
leading to $\bar E:=\dot\cup_{i\in N} E_i$ with $E_i=\{e^i_1,\dots e^i_m\}, i\in N$.
The domain of the integral polymatroid function $f_i$ is extended to $\bar E$ as follows
\[ \bar f_i(S):= f_i(E_i\cap S) \text{ for all }S\subset \bar E.\]
This way $\bar f_i(S)$ remains an integral  polymatroid rank function
on the lifted space $\Z^{n\cdot m}$.
Note that for $ \vec{\bar x}_i\in \bar \B_{f_i}$, we have $\vec{\bar x}_i\in  \Z^{n\cdot m}$ and with 
$f_i(\{\emptyset\})=0$, we get $ \bar x_{ij}=0$ for all $j\in \bar E\setminus E_i$. By this construction, we get $ \vec x_i \in \B_{f_i} \Leftrightarrow \vec{\bar x}_i\in \bar \B_{f_i}$.

Now we define the Minkowski sum
\[  \bar\B_1:= \sum_{i\in N} \bar \B_{f_i} \subset \Z^{n\cdot m},\]
which is again an integral  polymatroid base polyhedron.
By this construction we can represent all collections of 
integral base vectors by a single integral polymatroid base polyhedron.

It remains to also handle the capacity constraint~\eqref{ineq:load-polymatroid} (note that this is not a box constraint for
polymatroid $\bar\B_1$).
For $S\subset \bar E$, we define
$S_j:=\{j\in E \vert\;  \exists \; i\in N \text{ with } e^i_j\in S\} $ as the union of those original element indices (in $E$) for which $S$ contains at least one copy.
With this definition, we define a second polymatroid as follows.
\begin{align*}
\bar \B_2:=\{\vec{  x}\in \Z^{n\cdot m}\vert \;  x(S)\leq h(S) \text{ for all } S\subseteq \bar E, \;x(\bar E)= h(\bar E)\},
\end{align*}
where for $S\subset \bar E$
$ h(S):=\sum_{j \in S_j} u_j.$ One can easily verify 
that $h$ is an integral polymatroid function.
Now observe that for the sets
$\{e_1^j,\dots, e_n^j\}, j\in E$ we exactly get the
capacity constraint  $ x\left(\{e_1^j,\dots, e_n^j\}\right)\leq u_j, j\in E.$
Altogether, with the minimality of $\vec u$, problem~\ref{cp-polymatroid} can be  reduced to the problem of finding a  vector
in the intersection of $\bar \B_1$ and $\bar \B_2$ minimizing a
linear objective:
\begin{align}\label{eq:intersect}
\min\; \left\{ \sum_{i\in N}\sum_{j\in E}  \pi_{ij}(\vec u) x_{ij} \; \middle\vert \;\; 
\vec x \in \bar \B_1\cap \bar \B_2\right\}
\end{align}
By the fundamental result of Edmonds~\cite[Thm. (35)]{Edmonds2003},
the fractional relaxation $\bar {\EB}_1\cap \bar {\EB}_2$ is integral.
Note that there are strongly polynomial time
algorithms computing an optimal solution to~\eqref{eq:intersect} (see Cunningham and Frank~\cite{CunninghamF85} and Frank and Tardos~\cite{FrankT87}).

\end{proof}

\section{Utility Maximization on Polymatroids}\label{sec:polym-max}
For the maximization variant, the strategy
spaces $X_i, i\in N$ are usually defined as the vectors
of an integral polymatroid polyhedron $\P_{f_i}$.
We get the following reformulation of $P^{\max}(\vec u)$: 
 \begin{framed}
\begin{align*}\tag{$P^{\max-\polymatroid}(\vec u)$}
\max\;\sum_{i\in N}&\sum_{j\in E}  v_{ij}(\vec u) x_{ij}\\
\vec x_i & \in \P_{f_i} \text{ for all }i\in N\\
\ell_j(\vec x)&\leq u_j, j\in E
\end{align*}
\end{framed}
The following companion result for maximization problems on polymatroids
holds true.
\begin{theorem}\label{thm:max-polymatroid}
For polymatroid games, every $\vec u\in \Z^m$ 
for which $P^{\max}(\vec u)$ admits a finite optimal solution
 is weakly enforceable with market prices.
If $v_{ij}(\vec u)\geq 0$ for all $i\in N, j\in E$,
then, any supply vector $\vec u\in \Z^m_+$ for which there exists $\vec x\in X$ with $\ell(\vec x)=\vec u$ is enforceable.
\end{theorem}
\begin{proof}
The proof of the first statement  is analogous to the proof
of the previous theorem.
For the second statement,
with $v_{ij}(\vec u)\geq 0$ it follows that $v(\vec x)$ is monotonically nondecreasing. Moreover, for any integral $\vec u\in \Z^m_+$
for which there exists $\vec x\in X$ with $\ell(\vec x)=\vec u$ it is known
that $X$
is upwards closed -- using  polymatroid properties.
Thus, with the integrality of the polymatroid-intersection polytope (see the
proof of Theorem~\ref{thm:polymatroid-main}) the result follows.
\end{proof}

\end{document}